\documentclass[journal,10pt]{IEEEtran}

\usepackage[T1]{fontenc}		
\usepackage[utf8]{inputenc}	

\usepackage{cite}
\usepackage{graphicx}
\usepackage{amsmath,amssymb}
\usepackage{amsfonts}
\usepackage{amsthm}

\usepackage{hyperref}
\usepackage{subcaption}
\usepackage{nicefrac}
\usepackage{placeins} 
\usepackage[ruled,vlined,linesnumbered]{algorithm2e}
\usepackage[dvipsnames]{xcolor}

\newtheorem{remark}{Remark}
\newtheorem{theorem}{Theorem}

\newtheorem{prop}{Proposition}
\newtheorem{coro}{Corollary}

\usepackage{amsmath,lipsum}
\usepackage{cuted}
\usepackage{url}

\begin{document}
\title{On Secure Uplink Transmission in Hybrid RF-FSO Cooperative Satellite-Aerial-Terrestrial Networks}

\author{Yuanyuan Ma,~Tiejun~Lv,~\IEEEmembership{Senior Member,~IEEE}, Gaofeng Pan,~\IEEEmembership{Senior Member,~IEEE},~Yunfei~Chen,~\IEEEmembership{Senior Member,~IEEE},~and~Mohamed-Slim Alouini,~\IEEEmembership{Fellow,~IEEE}

\thanks{Manuscript received November, 29, 2021; revised April 6, 2022 and August, 11, 2022; accepted October 9, 2022. This work was supported by the National Natural Science Foundation of China (NSFC) under Grant 62271068 and 62171031. (\emph{corresponding author: Tiejun Lv}).}
\thanks{Y. Ma and T. Lv are with the School of Information and Communication Engineering, Beijing University of Posts and Telecommunications, Beijing 100876, China (e-mail: lvtiejun@bupt.edu.cn).}
\thanks{G. Pan is with the School of Cyberspace Science and Technology, Beijing Institute of Technology, Beijing 100081, China.}
\thanks{Y. Chen is with the School of Engineering, University of Warwick, Coventry CV4 7AL, U.K.}
\thanks{M.-S. Alouini is with the Computer, Electrical, and Mathematical Sciences and Engineering Division, King Abdullah University of Science and Technology (KAUST), Thuwal 23955-6900, Makkah Province, Saudi Arabia.}
}

{}

\maketitle

\begin{abstract}
This work investigates the secrecy outage performance of the uplink transmission of a radio-frequency (RF)-free-space optical (FSO) hybrid cooperative satellite-aerial-terrestrial network (SATN). Specifically, in the considered cooperative SATN, a terrestrial source (S) transmits its information to a satellite receiver (D) via the help of a cache-enabled aerial relay (R) terminal with the most popular content caching scheme, while a group of eavesdropping aerial terminals (Eves) trying to overhear the transmitted confidential information. Moreover, RF and FSO transmissions are employed over S-R and R-D links, respectively. Considering the randomness of  R, D, and Eves, and employing a stochastic geometry framework, the secrecy outage performance of the cooperative uplink transmission in the considered SATN is investigated and a closed-form analytical expression for the end-to-end secrecy outage probability is derived. Finally, Monte-Carlo simulations are shown to verify the accuracy of our analysis.
\end{abstract}


\IEEEpeerreviewmaketitle

\section{Introduction}
\global\long\def\figurename{Fig.}

Satellite communication has attracted great attention due to its inherent characteristics, e.g., high capability of seamless connectivity and wide coverage \cite{Guo_SOP,Power_Control_Cognitive_Satellite,Lin_BF_rate}.
More recently, it is becoming an important enhancement of the six-generation (6G) systems to support the exponentially increasing data demand and variety of users across the world, since satellite communication can be widely applied in mass broadcasting, navigation, and disaster relief operations \cite{6G1,6G2,Ye4,Ye5}.

However, direct communication links between the satellite and the terrestrial terminals may not always be  available, due to deep fading (e.g., the shadowing created by buildings and mountains) \cite{Pan_sate}.  Thus, aerial relays, like unmanned aerial vehicle (UAV) and high/low altitude platforms, have been regarded as an alternative and promising solution to extend and improve satellite-terrestrial communications. One of the advantages of aerial relays is their flexibility for quick deployment. Thus they have been used to substitute for traditional on-ground base stations, because of emergency or lack of terrestrial infrastructure coverage \cite{secure_Tingting_LI}.
The cooperative satellite-aerial-terrestrial network (SATN), which can effectively mitigate the impacts of deleterious masking effect in satellite links, has attracted a significant amount of attention \cite{Tian_Yu,UAV_pan,Pan_sate,Xiang_Li_UAV,Satellite-Aerial_Huang,Lin_IOT}. 
Reference \cite{Tian_Yu} derived the coverage probability of a dual-hop cooperative satellite-UAV communication system.
Reference  \cite{UAV_pan} considered a cooperative satellite-aerial-terrestrial system and derived the
approximate analytical expressions for the coverage probability over the relay-terrestrial-receiver link and the
end-to-end
outage probability (OP) in both non-interference and interference scenarios.
The UAV trajectory and in-flight transmit power were jointly optimized by using a typical composite channel model including both large-scale and small-scale fading  \cite{Xiang_Li_UAV}.
The energy efficiency of the considered cooperative SATN with beamforming schemes was analytically presented in \cite{Satellite-Aerial_Huang}, where a multi-antenna UAV is employed as a relay to assist the satellite signal delivery.
To maximize the sum rate of cooperative SATN  with rate-splitting multiple access, \cite{Lin_IOT} proposed an iterative penalty function-based algorithm to support massive access of Internet-of-Things devices and achieve the desired performances of interference suppression, spectral efficiency, and hardware complexity.

Furthermore, wireless caching \cite{caching_wire} has emerged as an effective paradigm where some contents are prefetched locally by the network nodes in their installed storage during off-peak
hours.  Generally, there are two fundamental caching schemes, namely most popular content (MPC) caching and uniform content (UC) caching \cite{caching_wire}. While the UC caching achieves the largest
content diversity gain,  the MPC caching achieves the largest cooperative diversity gain \cite{caching_wire,Sharma_uav_cach,An_2021}. The
end-to-end
 OP of a cooperative SATN was evaluated in \cite{Sharma_uav_cach}, considering the fundamental MPC and UC caching schemes at UAV relays.
The OP and hit probability of the cache-enabled cooperative SATN were derived in \cite{An_2021}, taking into account the uncertainty of the number and location of the aerial node. Compared to the  UC scheme, the MPC scheme is widely used with a high hit rate.

On the other hand, the information security is important in wireless  communication due to the open access in wireless mediums.
There are extensive works that have characterized SOP and ergodic secrecy capacity (ESC) for passive eavesdropping case and active eavesdropping case \cite{Pan_sate,secure_Tingting_LI,Tian_Yu,UAV_pan,ESC_JXue}, where passive eavesdropping is the most common in practice.
Moreover, free-space optical (FSO) links have been presented as an ideal alternative to the conventional radio frequency (RF) links for secure satellite systems, because the laser beam has high directionality for security \cite{laser_lei}. By utilizing relaying technology, the hybrid RF-FSO systems combine both the advantages of the RF and FSO communication technologies \cite{HJLei_fso,FSO_RF_TWC,RF_FSO_open}. Dual-hop  RF-FSO systems are designed to overcome atmospheric turbulence and other factors limiting the applications of FSO systems \cite{HJLei_fso,FSO_RF_TWC} and performances of RF-FSO systems in term of
end-to-end
OP, average symbol error rate, and ergodic capacity (EC), have already  been vastly studied and analyzed \cite{RF_FSO_open,Zedini_TWC,Varshney}. The physical layer security of the  RF-FSO systems was investigated in \cite{laser_lei,HJLei_fso}. \cite{laser_lei} analyzed secrecy outage performance of a mixed RF-FSO transmission system with imperfect channel state information (CSI), where an eavesdropper wiretaps the confidential information by decoding the received signal.  The secrecy outage performance of a hybrid RF-FSO downlink system was analyzed in \cite{HJLei_fso} where energy harvesting technology was considered over RF links.


Most of the authors of the aforementioned works focus on the downlink transmission performance of satellite-terrestrial/SATN systems, while the uplink transmission performance of cooperative SATN has not been extensively studied.
However, the study of the security of the uplink transmission is also important and meaningful, due to the following reasons.

Firstly,   in real-time communication scenarios, satellites fundamentally play the role of a space relay/forwarder to connect two or more remote terrestrial terminals, while in non real-time communication scenarios, satellites usually collect data reported and transmitted by the sensors at any time  and forward the  collected data to data center using dedicated channel \cite{9475123}. This leads to a fact that the information security over the uplink is equally important or more worthy of study. Secondly, thinking of the inherent openness of wireless medium, the three-dimensional coverage space of the uplink transmission is much larger than that of the downlink transmission, resulting from the cut of  the floor plane or not. It means that, compared with the downlink transmission, there is a larger distribution space for the eavesdroppers in the uplink transmission scenarios, resulting in a tougher challenge to shield the information delivered over the uplink.

%
%

Motivated by these observations, we built an  RF-FSO hybrid cooperative SATN with multiple eavesdropping aerial terminals (Eves) and one cache-enabled aerial relay (R), in which RF and FSO transmissions are employed over terrestrial source (S)-R and R-satellite receiver (D) links, respectively, to exploit their inherent merits, e.g., the wide capability provided by RF transmission and the highly directive nature of laser beam offered by FSO techniques.
Differing from conventional studies on downlink transmissions, we investigate the outage performance of the uplink transmission in the considered SATN.
Moreover, considering the randomness of the positions of  R, D, and Eves, modeled by using stochastic geometry, the closed-form expression for the end-to end secrecy outage probability (SOP) of the cooperative uplink transmission in the considered SATN is obtained.
In particular, the investigation of the randomness of the satellite's position has practical reason and great significance.
With the  more and more intensive distribution of satellites and under the condition of unknown satellites' trajectories   and distributions, the connectable satellites are assumed to be
uniformly distributed in  a part of the spherical cone with known altitude range \cite{Pan_sate} or
uniformly distributed on a spherical doom with known altitude, which is called as randomness of the satellite's position.
On one hand, the study on the randomness of the satellite's position can offer the secrecy performance  with unknown trajectories and distributions of satellites.
On the other hand, the randomness analysis can give some hints with known satellite constellation.

The main contributions of this paper are summarized as follows:
\begin{itemize}
\item  We built an  RF-FSO hybrid cooperative SATN with multiple Eves and a cache-enabled R, and investigated the outage performance of the uplink transmission in the considered SATN.
\item We derived a closed-form  expression for the lower bound of the end-to-end SOP over the cooperative uplink in the considered SATN,  considering the small scale fading of two types of channels, the cache-enabled R, and the randomness of the positions of R, D, and Eves.
\item  The impacts of the system factors,  including the channel fading parameters,  the transmit power and the number of antennas at R, etc, have been studied and discussed,  which can facilitate researchers to design systems considering both reliability and exiting constrains.
\end{itemize}

The rest of this paper is organized as follows. In Section II, the considered hybrid RF-FSO cooperative SATN is described. Section III and IV analyze the SOP of S-R link and the OP of R-D link, respectively. In Sections V, the end-to-end SOP of the RF-FSO cooperative SATN is analyzed. In Section VI, numerical results are presented and discussed. Finally, Section VII concludes the paper.

\section{System Descriptions}
\subsection{System Model}
\begin{figure}[!h]
\setlength{\abovecaptionskip}{0pt}
\setlength{\belowcaptionskip}{10pt}
\centering
\includegraphics[width=3.0 in]{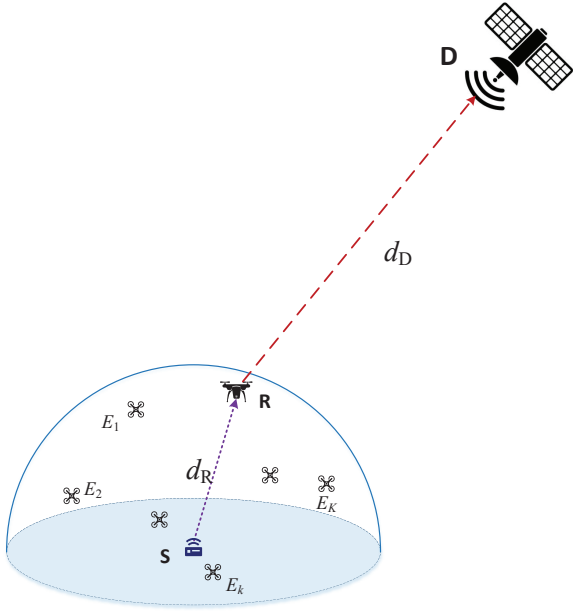}
\caption{The uplink transmission of the RF-FSO hybrid cooperative SATN with a cache-enabled R and
spatially random eavesdroppers}
\label{system}
\end{figure}

Consider a hybrid RF-FSO cooperative SATN, which consists of a terrestrial source (S), a cache-enabled aerial relay (R), a satellite receiver (D), and a group of aerial eavesdroppers ($\mathrm E_{k},1\leq k\leq K$), as shown in Fig. \ref{system}. Specifically, the S-R and the S-Eves links with RF transmission experience independent and identical Nakagami-$m$ fading, while the R-D link with FSO transmission follows a unified Gamma-Gamma fading. S transmits its information to D via the help of R with the MPC caching scheme, while a group of Eves try to overhear the transmitted confidential information.

Here, we assume that R is equipped with $L \ge 2$  antennas and that maximum ratio combining (MRC) scheme is employed to process the received signals to achieve the maximum instantaneous signal-to-noise ratio (SNR), while S and each Eve are all equipped with a single antenna for simplicity\footnote{
When S is equipped with multiple antennas and the MRC scheme is implemented at R,
 a best antenna can be chosen to transmit the information bits to achieve the best SNR by using transmit antenna selection (TAS) scheme \cite{TAS_MRC}. Moreover, the considered eavesdropping scenario with multiple single-antenna eavesdroppers is actually equivalent to the case of a multi-antenna eavesdropper.}. Furthermore, the omnidirectional transmission antenna is assumed to be employed at S.

\subsection{Channel Model for RF/FSO Link}
\textit{1) S-R RF Link}\\
The fading amplitudes of links $\mathrm S\rightarrow \mathrm R_{l}$, $S\rightarrow \mathrm E_{k}$,
which describe the channel fading between $\mathrm S$ and the $l$-th antenna of $\mathrm R$, $\mathrm S$ and the $k$-th  Eve, are denoted by
$h_q$, where $q=\left\{ \mathrm{SR}_{l},\mathrm{SE}_{k}\right\}$.

Consequently, the channel power  $g_{q}=\left|h_{q}\right|^{2}$ are Gamma distributed with probability density function (PDF) and cumulative density function (CDF)
\begin{align}\label{Nakagami_PDF}
f_{g_{q}}\left(x\right)=\frac{\lambda_{q}^{m_{q}}}{\Gamma\left(m_{q}\right)}x^{m_{q}-1}\exp\left(-\lambda_{q}x\right)
\end{align}
and
\begin{align}\label{Nakagami_CDF}
F_{g_{q}}=\frac{\gamma\left(m_{q},\lambda_{q}x\right)}{\Gamma\left(m_{q}\right)}
,
\end{align}
respectively, where $\lambda_{q}=\frac{m_{q}}{\Omega_{q}}$, $m_{q}$ and $\Omega_{q}$ denote the fading severity and the average channel power, respectively, $\Gamma\left(.\right)$ and $\gamma\left(.,.\right)$ are
the Euler Gamma function
\cite[Eq. (8.310.1)]{Gradshteyn}  and the lower incomplete Gamma function \cite[Eq. (8.350.1)]{Gradshteyn}, respectively. For an integer $m_{q}$, \eqref{Nakagami_CDF}
can be written as \cite[Eq. (8.352.1)]{Gradshteyn}
\begin{align}\label{g_q}
F_{g_{q}}\left(x\right)=1-\exp\left(-\lambda_{q}x\right)\sum_{k=0}^{m_{q}-1}\frac{\lambda_{q}^{k}x^{k}}{k!}
.
\end{align}
We also assume that the channels between S and each antenna of
R,  channels between  S and each Eve experience independent Nakagam-$m$ fading. For simplicity, let $m_\mathrm{R}$ and $\Omega_\mathrm{R}$ respectively denote the fading severity and the average channel power between S and each  antenna of R, $m_\mathrm{E}$ and $\Omega_\mathrm{E}$ respectively denote the fading severity and the average channel power between S and each  Eve.

Meanwhile, when the MRC scheme is implemented at $R$, for a $(1,L)$ MRC system with a single transmit antenna and $L$ receive antennas in Nakagam-$m$ fading channels, the PDF and CDF of the combined channel power $h_\mathrm{SR}$ can be shown as \cite{MRC_Nakagami_1,MRC_Nakagami_2}
\begin{align}\label{MRC_Nakagami_PDF}
f_{\left\Vert h_\mathrm{SR}\right\Vert ^{2}}\left(x\right)=\frac{\lambda_\mathrm{R}^{Lm_\mathrm{R}}}{\Gamma\left(Lm_\mathrm{R}\right)}x^{Lm_\mathrm{R}-1}\exp\left(-\lambda_\mathrm{R}x\right)
\end{align}
and
\begin{align}\label{MRC_Nakagami_CDF}
F_{\left\Vert h_\mathrm{SR}\right\Vert ^{2}}\left(x\right)&=\frac{\gamma\left(Lm_\mathrm{R},\lambda_\mathrm{R}x\right)}{\Gamma\left(Lm_\mathrm{R}\right)}\notag\\
&
=1-\exp\left(-\lambda_\mathrm{R}x\right)\sum_{k=0}^{Lm_\mathrm{R}-1}\frac{\lambda_\mathrm{R}^{k}x^{k}}{k!}
,\end{align}
respectively, where $\lambda_\mathrm{R}=\frac{m_\mathrm{R}}{\Omega_\mathrm{R}}$.

\textit{2) R-D FSO Link\\ }
In this work, the FSO communication link between R and D follows the Gamma-Gamma distribution, which accounts for pointing errors and type of detection techniques \cite{Gamma-Gamma,FSO}. The PDF $f_{\gamma_\mathrm{D}}\left(x\right)$ and CDF $F_{\gamma_\mathrm{D}}\left(x\right)$ of the instantaneous SNR at D  $\gamma_\mathrm{D}$ are given as \cite{HJLei_fso}
\begin{align}
f_{\gamma_\mathrm{D}}\left(x\right)=Ax^{-1}G_{1,3}^{3,0}\left[Bx^{\frac{1}{r}}|\begin{array}{c}
\xi^{2}+1\\
\xi^{2},a,b
\end{array}\right]
\end{align}
and
\begin{align}\label{CDF_FSO}
F_{\gamma_\mathrm{D}}\left(x\right)=IG_{r+1,3r+1}^{3r,1}\left[\rho x|\begin{array}{c}
1,K_{1}\\
K_{2},0
\end{array}\right],
\end{align}
respectively, where
$A=\frac{\xi^{2}}{r\Gamma\left(a\right)\Gamma\left(b\right)}$, $B=\frac{hab}{\sqrt[r]{\Omega_\mathrm{D}}}$,
$I=\frac{\xi^{2}r^{a+b-2}}{\left(2\pi\right)^{r-1}\Gamma\left(a\right)\Gamma\left(b\right)}$,
$\rho=\frac{\left(hab\right)^{r}}{\Omega_\mathrm{D}r^{2r}}$, $K_{1}=\Delta\left(r,\xi^{2}+1\right)$,
$K_{2}=\left[\Delta\left(r,\xi^{2}\right),\Delta\left(r,a\right),\Delta\left(r,b\right)\right]$,
in which the parameters $a$ and $b$ are used to represent the severity
of fading/scintillation due to the atmospheric turbulence conditions, $r$ represents the
detection scheme used at D, i.e. $r = 1$ for heterodyne detection (HD) and $r = 2$ for  intensity modulation with direct detection (IM/DD),
$\xi$ is the ratio of the equivalent beam radius to the standard deviation of the pointing
error displacement (jitter) at the FSO receiver \cite{FSO_RF_TWC}, $\Omega_\mathrm{D}$
represents the average electrical SNR of the FSO link, $\Delta\left(k,a\right)=\frac{a}{k},\frac{a+1}{k},\cdots,\frac{a+k-1}{k}$,
$h=\frac{\xi^{2}}{\xi^{2}+1}$, and $G_{g,q}^{m,n}\left[\cdot\right]$
is Meijer\textquoteright s G-function, as defined by  \cite[Eq. (9.301)]{Gradshteyn}.

\subsection{Signal Model}
Since the modulation schemes over RF and FSO links may not be the same in most cases, here decode-and-forward (DF) relaying  scheme is considered at R \cite{V2V_TVT}.

In the first phase, at time $t$, S transmits its information bits $x_\mathrm s (t)$ (satisfying $\mathbf{E}\left\{ \left\Vert x_\mathrm{s}\left(t\right)^{2}\right\Vert \right\} =1$) to R. Then, the received signal at R and $E_k$
can be given as
\begin{align}
y_\mathrm{R}\left(t\right)=\mathbf{h}_\mathrm{SR}\sqrt{P_\mathrm{S}d_\mathrm{R}^{-\eta_1}}x_\mathrm s\left(t\right)+n_\mathrm{R}
\end{align}
and
\begin{align}
y_{\mathrm E k}\left(t\right)={h}_{\mathrm{SE}_{k}}\sqrt{P_{S}d_{\mathrm E k}^{-\eta_1}}x_\mathrm s\left(t\right)+n_\mathrm{E},
\end{align}
where
$P_\mathrm{S}$ is the transmit power at S,
$n_\mathrm{R} \sim \mathcal{CN}(0,N_\mathrm R)$ and $n_\mathrm{E} \sim \mathcal{CN}(0,N_\mathrm E)$ are the Gaussian noise at R and E, $d_\mathrm{R}$ and $d_{\mathrm E k}$ denote the distance between  S to  R, S to  the $k$-th E, respectively, $\eta_1>0$ denotes the path-loss factor.

Thus, the SNR at R and  $\mathrm E_k$ can be written as
\begin{align}\label{gamma_r}
\gamma_\mathrm{R}=\frac{P_\mathrm{S}\left\Vert \mathbf{h}_\mathrm{SR}\right\Vert ^{2}}{N_\mathrm{R}d_\mathrm{R}^{\eta_1}}
\end{align}
and
\begin{align}
\gamma_{\mathrm E k}=\frac{P_\mathrm{S}\left\Vert {h}_{\mathrm{SE}_{k}}\right\Vert ^{2}}{N_\mathrm{E}d_{\mathrm Ek}^{\eta_1}},
\end{align}
respectively.
If the $\gamma_\mathrm{R}$ is greater than a threshold $\gamma_{\mathrm{hold}}$ ($\gamma_{\mathrm{hold}} > 0$) \cite{DF_threshold}, R successfully decodes the signal $x_\mathrm{s}$, otherwise,  R  will not be able to successfully decode the signal $x_\mathrm{s}$.

Considering large-scale path-loss and symmetric fading channels, the strongest Eve is the Eve nearest to S among all $K$ Eves. The shortest distance between S and Eves is denoted by $d_\mathrm E=\min\{d_{\mathrm E 1}, d_{\mathrm E 2}, \cdots, d_{\mathrm E K}\}$.
The SNR at the strongest Eve, denoted as $\mathrm{E}^*$, is given by
\begin{align}\label{gamma_e}
\gamma_\mathrm{E}=\frac{P_\mathrm{S}\left\Vert {h}_{\mathrm{SE}^{*}}\right\Vert ^{2}}{N_\mathrm{E}d_\mathrm{E}^{\eta_1}},
\end{align}
where ${h}_{\mathrm{SE}^{*}}$ is the channel fading between S and the strongest Eve.

In the second phase, R sends optical signal $x_\mathrm r$, which is the correctly decoded signals from S and satisfies $\mathbf{E}\left\{ \left\Vert x_\mathrm{r}\left(t\right)^{2}\right\Vert \right\} =1$, to D through the FSO link, where the pointing loss and scattering loss are not considered \cite{fso_DL_TWC}. Meanwhile, after conducting optical-to-electrical conversion with photo detector, the output electrical signal at D can be expressed as
\begin{align}
y_\mathrm{D}\left(t\right)=\zeta\sqrt{\frac{P_\mathrm{R}}{\mathcal{L}_{\mathrm{FS}}}}\mathcal{L}_\mathrm{r}I_{\mathrm{fso}}x_\mathrm{r}\left(t\right)
+n_\mathrm{D}\left(t\right)
,
\end{align}
where $P_\mathrm{R}$ denotes the transmit power at R and $\zeta$ denotes the optical-to-electrical conversion coefficient,
 $\mathcal{L}_{\mathrm{FS}}=\left(\frac{4\pi f_\mathrm c d_\mathrm{D}}{c}\right)^{2}$
is the free space path loss, in which $d_{D}$ is the distance between R and D, $c$ is the velocity
of light and $f_c$ is the carrier wavelength of transmitted signal from R,
while $\mathcal{L}_\mathrm{r}=\frac{1}{2}\left(G_\mathrm{t}+G_\mathrm{r}-A_\mathrm{Atm}-A_\mathrm{lenses}-A_\mathrm{mar}\right)$
$[\mathrm{dB}]$ with  $G_\mathrm{t}$, $G_\mathrm{r}$, $A_\mathrm{Atm}$, $A_\mathrm{lenses}$
and $A_\mathrm{mar}$ being, respectively, transmitter gain, receiver gain, atmospheric attenuation, lenses losses, and system margin \cite{fso_DL_TWC}.  Besides, $I_{\mathrm{fso}}$ represents
the channel fading coefficient of the FSO link and $n_\mathrm{D}\left(t\right)$ represents additive white Gaussian noise (AWGN) with zero mean and
variance $\sigma_\mathrm{d}^{2}$ at D. As a result, the instantaneous SNR of the FSO link can be expressed as
\begin{align}
\gamma_\mathrm{D}=\frac{P_\mathrm{R}\zeta^{2}\mathcal{L}_\mathrm{r}^{2}I_{\mathrm{fso}}^{2}}{\mathcal{L}_{\mathrm{FS}}\sigma_\mathrm{d}^{2}}
\overset{\Delta}{=}\Omega_{D}I_{\mathrm{fso}}^{2},
\end{align}
where
\begin{align}\label{rs_ave}
\Omega_\mathrm{D}=\frac{P_{R}\zeta^{2}\mathcal{L}_\mathrm{r}^{2}}{\mathcal{L}_{\mathrm{FS}}\sigma_\mathrm{d}^{2}}
\end{align}
is the average received SNR at D.

The secrecy performance of the considered cooperative SATN is based on the  SOP over the S-R link in the presence of Eves, and  the OP over the R-D link, all of which are important and indispensable. Next we will analyze them in Sec. III and Sec IV, respectively.

\section{SOP Analysis for S-R Link}
\subsection{Preliminaries}
As depicted in Fig. 2, the coverage space of S is a hemisphere with the center $\mathrm{O}$ and radius $R_\mathrm S$, when the omnidirectional transmission antenna is assumed to be employed at S. To address the randomness of the positions of Eves, we put S at the origin O and it is assumed that all Eves are uniformly distributed in the coverage space of S  to eavesdrop the information delivery between S and R.
Therefore, the PDF and CDF of the distance between S and ${E_k}$ can be written as \cite{UAV_pan}
\begin{align}
f_{d_{\mathrm E k}}\left(x\right)=\begin{cases}
\frac{3x^{2}}{R_{S}^{3}}, &\mathrm{if}~0\leq x\leq R_\mathrm{S};\\
0, & \mathrm{else}
\end{cases}
\end{align}
and
\begin{align}
F_{d_{\mathrm E k}}\left(x\right)=\begin{cases}
0, & \mathrm{if}~x<R_\mathrm{S};\\
\frac{x^{3}}{R_\mathrm{S}^{3}}, &\mathrm{if}~0\leq x\leq {R_\mathrm{S}};\\
1, & \mathrm{else}
\end{cases},
\end{align}
respectively.

Further, the CDF and the PDF of $d_\mathrm E$ can also be written as
\begin{align}
{F_{{d_\mathrm{E}}}}\left( {{d_\mathrm E}} \right) &= \Pr \left\{ {\min \left\{ {{d_1},{d_2}, \cdots ,{d_K}} \right\} \le {d_\mathrm E}} \right\} \notag\\
& = 1 - \Pr \left\{ {\min \left\{ {{d_1},{d_2}, \cdots ,{d_K}} \right\} > {d_\mathrm E}} \right\} \notag\\
& = 1 - \Pr \left\{ {{d_1} > {d_\mathrm E}} \right\} \cdots \Pr \left\{ {{d_K} > {d_\mathrm E}} \right\} \notag\\
& = 1 - {\left( {1 - \frac{d_\mathrm E^{3}}{R_\mathrm{S}^{3}}} \right)^K}
\end{align}
and
\begin{align}\label{dE_min_PDF}
{f_{{d_\mathrm E}}}\left( {{d_\mathrm E}} \right) = \frac{{\partial {F_{{d_\mathrm E}}}\left( {{d_\mathrm E}} \right)}}{{\partial {d_\mathrm E}}} = K{\left( {1 - \frac{d_\mathrm E^{3}}{R_{S}^{3}}} \right)^{K - 1}}\frac{{3{d_E}^2}}{{{R_\mathrm S^3}}},
\end{align}
respectively, where $0 \le {d_\mathrm E} \le R_\mathrm S$.

By using Jacobian matrix method \cite{Papoulis}, the PDF of $d_\mathrm E^{\eta_1}$ is derived from \eqref{dE_min_PDF} as
\begin{align}\label{PDF_Z}
{f_{d_\mathrm E^{\eta_1}}}\left( x \right) &= {f_{{d_\mathrm E}}}\left( {{x^{  {1 \mathord{\left/
 {\vphantom {1 \eta }} \right.
 \kern-\nulldelimiterspace} \eta_1 }}}} \right){\left| {\frac{{\partial d_\mathrm E^{\eta_1}}}{{{d_\mathrm E}}}} \right|^{ - 1}} \notag\\
 &\mathop  = \limits^{\left( a \right)}  \sum_{f=1}^{K}{K \choose f}\frac{3f\left(-1\right)^{f+1}}{\eta_1 R_\mathrm{S}^{3f}}x^{\frac{3f}{\eta_1}-1},
\end{align}
where $(a)$ follows a binomial expansion.

\begin{figure}[!h]
\setlength{\abovecaptionskip}{0pt}
\setlength{\belowcaptionskip}{10pt}
\centering
\includegraphics[width=3.2 in]{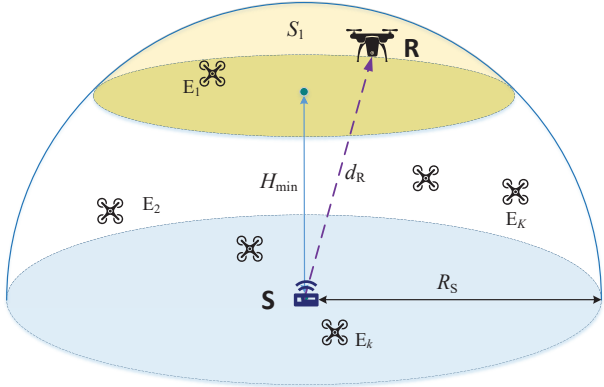}
\caption{The S-R link}
\label{SR}
\end{figure}

Further, to reflect the practical operation scenarios, we assume that R operates in the coverage space of S with minimum height $H_{\mathrm{min}}$ above the ground plane where $H_{\mathrm{min}} \leq R_\mathrm S$.
As shown in  Fig. \ref{SR},  R is assumed to be uniformly distributed
in the spherical cap $S_{1}$ with the radius $r_{\mathrm{cap}}=\sqrt{R_\mathrm{S}^{2}-H_{\mathrm{min}}^{2}}$
of the base of the cap and height $h_{\mathrm{cap}}=R_\mathrm{S}-H_{\mathrm{min}}$
. Then, it is easy to have the volume of the spherical cap $S_{1}$ as
\begin{align}
V_{\mathrm S_{1}}=\frac{\pi}{3}\left(2R_\mathrm{S}^{3}-3H_{\mathrm{min}}R_\mathrm{S}^{2}+H_{\mathrm{min}}^{3}\right).
\end{align}

In order to facilitate the following analysis, spherical coordinates
are adopted. Then, the coordinate
of R can be presented as $\left(r_\mathrm{R},\theta_\mathrm{R},\psi_\mathrm{R}\right)$,
where $H_{\mathrm{min}}\le r_\mathrm{R}\le R_\mathrm{S}$, $0\le\theta_\mathrm{R}\le\arccos\frac{H_{\mathrm{min}}}{R_\mathrm{S}}$
and $0\le\psi_\mathrm{R}\le2\pi$.

Therefore, employing Lemma 4 of \cite{UAV_pan1}, the CDF of the distance between S and R, $d_\mathrm{R}=r_\mathrm{R}$, can be derived as
\begin{align}
F_{d_\mathrm{R}}\left(x\right)&=\frac{1}{V_{S_{1}}}\intop_{0}^{\arccos\left(\frac{H_{\mathrm{min}}}{x}\right)}\intop_{\frac{H_{\mathrm{min}}}{\cos\theta}}^{x}\intop_{0}^{2\pi}\sigma^{2}\sin\left(\theta\right)d\psi d\sigma d\theta
 \notag\\
&
=\frac{2\pi}{V_{S_{1}}}\intop_{0}^{\arccos\left(\frac{H_{\mathrm{min}}}{x}\right)}\intop_{\frac{H_{\mathrm{min}}}{\cos\theta}}^{x}\sigma^{2}\sin
\left(\theta\right)d\sigma d\theta
 \notag\\
&
=\frac{2\pi}{3V_{S_{1}}}\intop_{0}^{\arccos\left(\frac{H_{\mathrm{min}}}{x}\right)}\left(x^{3}-H_{\mathrm{min}}^{3}\sec^{3}\left(\theta\right)\right)
\sin\left(\theta\right)d\theta
 \notag\\
&
=\frac{\pi}{3V_{S_{1}}}\left(2x^{3}-3H_{\mathrm{min}}x^{2}+H_{\mathrm{min}}^{3}\right).
\end{align}

So, the PDF of $d_\mathrm{R}$ and ${d_\mathrm{R}^{\eta_1}}$ can be obtained as
\begin{align}\label{cdf_dR}
f_{d_\mathrm{R}}\left(x\right)=\frac{2\pi}{V_{S_{1}}}\left(x^{2}-H_{\mathrm{min}}x\right)
\end{align}
and
\begin{align}
f_{d_\mathrm{R}^{\eta_1}}\left(x\right)=\frac{2\pi}{\eta_1 V_{  S_{1}}}\left(x^{\frac{3}{\eta_1}-1}-H_{\mathrm{min}}x^{\frac{2}{\eta_1}-1}\right),
\end{align}
respectively.

\begin{prop}
The PDF of $Z=\frac{d_\mathrm{R}^{\eta_1}}{d_\mathrm{E}^{\eta_1}}$  is expressed on the top of next page shown in \eqref{dZ},
\begin{figure*}
\begin{align}\label{dZ}
f_{Z}\left(z\right)=\begin{cases}
\sum_{f=1}^{K}A_{f}z^{-\frac{3f}{\eta_1}-1}, &\mathrm{if}~z>1;\\
\sum_{f=1}^{K}\left(B_{1,f}z^{\frac{3}{\eta_1}-1}+B_{2,f}z^{\frac{2}{\eta_1}-1}+B_{3,f}z^{-\frac{3f}{\eta_1}-1}\right), & \mathrm{if}~ \mathrm{\mathit{\frac{H_{\mathrm{min}}^{\eta_1}}{R_\mathrm{D}^{\eta_1}}\leq z\leq}1};\\
0, & \mathrm{else}
\end{cases}
\end{align}
\rule{18cm}{0.01cm}
\end{figure*}
where
\begin{align}
A_{f}=A_{1,f}+A_{2,f}
,\end{align}
\begin{align}\label{B1f}
B_{1,f}={K \choose f}\frac{6\pi}{\eta_1 V_{S_{1}}}\frac{\left(-1\right)^{f+1}f}{R_\mathrm{S}^{3f}}\frac{R_\mathrm{S}^{3f+3}}{3f+3}
,\end{align}

\begin{align}\label{B2f}
B_{2,f}=-{K \choose f}\frac{6\pi}{\eta_1 V_{S_{1}}}\frac{\left(-1\right)^{f+1}f}{R_\mathrm{S}^{3f}}\frac{H_{\mathrm{min}}R_\mathrm{S}^{3f+2}}{3f+2}
,\end{align}
and
\begin{align}\label{B3f}
B_{3,f}=&{K \choose f}\frac{6\pi}{\eta_1 V_{S_{1}}}\frac{\left(-1\right)^{f+1}f}{R_\mathrm{S}^{3f}}H_{\mathrm{min}}^{3f+3}
\notag\\
&\times
\left(\frac{f}{3f+3}-\frac{3f}{2\left(3f+2\right)}+\frac{1}{6}\right)
,\end{align}
in which $A_{1,f}$ and $A_{2,f}$ can be written as
\begin{align}\label{A1f}
A_{1,f}={K \choose f}\frac{\pi}{\eta_1 V_{S_{1}}}\left(-1\right)^{f+1}f\left(2R_\mathrm{S}^{3}-3H_{\mathrm{min}}R_\mathrm{S}^{2}+H_{\mathrm{min}}^{3}\right)
\end{align}
and
\begin{align}\label{A2f}
A_{2,f}=&{K \choose f}\frac{\pi}{\eta_1 V_{ S_{1}}}\left(-1\right)^{f}\frac{f^{2}}{R_\mathrm{S}^{3f}}\left[\frac{2}{f+1}R_\mathrm{S}^{3f+3}\right.
\notag\\&
-\frac{9}{3f+2}R_{S}^{3f+2}H_{\mathrm{min}}+\frac{1}{f}R_\mathrm{S}^{3f}H_{\mathrm{min}}^{3}
\notag\\&
\left.
+\left(\frac{9}{3f+2}-\frac{2}{f+1}-\frac{1}{f}\right)H_{\mathrm{min}}^{3f+3}\right ]
,\end{align}
respectively.

\begin{proof}
See Appendix A.
\end{proof}
\end{prop}

\subsection{SOP Analysis for S-R Link}
The secrecy capacity for S-R link in condition that R can successfully decode the signal $x_\mathrm s$ is defined as \cite{Bloch}
\begin{align}\label{secrecy_capacity_def}
{C_\mathrm S} (\gamma_\mathrm R, \gamma_\mathrm E) = \max \left\{ {{{\log }_2}\left( {1 + {\gamma _\mathrm R}} \right) - {{\log }_2}\left( {1 + {\gamma _\mathrm E}} \right),0} \right\}.
\end{align}

In this work, passive eavesdropping is assumed to reflect the
most common eavesdropping scenario for the eavesdropper to
achieve the best eavesdropping and to keep itself from being
uncovered \cite{secure_Tingting_LI}. In other words, S has no CSI of the eavesdropping
channel and SOP is investigated.
Thus, the SOP in the first phase is thus defined as the probability that the secrecy capacity is below a certain threshold ($C_\mathrm{th}$) with R  decoding successfully \cite{secure_Tingting_LI}, which is given by
\begin{align}\label{sop_1}
{\rm{SOP}_1} &= \Pr \left\{ {{{\log }_2}\left( {1 + {\gamma _\mathrm R}} \right) - {{\log }_2}\left( {1 + {\gamma _\mathrm E}} \right) \le {C_\mathrm{th}}} \right\} \notag\\
& = {\Pr \left\{ {{\gamma _\mathrm R} \le \lambda {\gamma _\mathrm E} + \lambda  - 1} \right\}}
\notag\\
& \geq {\Pr \left\{ {{\gamma _\mathrm R} \le \lambda {\gamma _\mathrm E} } \right\}}
,
\end{align}
where $\lambda=2^{C_\mathrm{th}}$. Here we assume that $\gamma_{\mathrm{hold}}<\lambda {\gamma _\mathrm E}$ which means that a secure transmission would includes the case of successfully decoding at R.

\begin{figure*}
\begin{align}\label{sop1_final}
\mathrm{SOP}_{1}^{\mathrm{L}}=&1-\frac{\lambda_\mathrm{E}^{m_\mathrm{E}}}{\Gamma\left(m_\mathrm{E}\right)}\sum_{k=0}^{Lm_\mathrm{R}-1}\sum_{f=1}^{K}
\frac{a_{0}^{k}}{k!}\left[
B_{1,f}H_{1}\left(\rho,\lambda_\mathrm{E},a_{0},k+m_\mathrm{E}-1,k+\frac{3}{\eta_1}-1\right)\right.
\notag\\&+B_{2,f}H_{1}\left(\rho,\lambda_\mathrm{E},a_{0},k+m_\mathrm{E}-1,k+\frac{2}{\eta_1}-1\right)+B_{3,f}H_{1}\left(\rho,
\lambda_\mathrm{E},a_{0},k+m_\mathrm{E}-1,k-\frac{3f}{\eta_1}-1\right)
\notag\\&\left.
+A_{f}H_{2}\left(\lambda_\mathrm{E},a_{0},k+m_\mathrm{E}-1,k-\frac{3f}{\eta_1}-1\right)
\right]
\end{align}
\rule{18cm}{0.01cm}
\end{figure*}
\begin{theorem}\label{SOP1_theo}
The lower bound of  $\mathrm{SOP}_{1}$ for S-R link  of the considered  RF-FSO cooperative SATN can be derived as \eqref{sop1_final} shown on the top of next page, where
$a_{0}=\lambda\lambda_\mathrm{R}\frac{N_\mathrm{R}}{N_\mathrm{E}}$, the functions $H_{1}\left(\varrho,a,b,q,p\right)$  and $H_{2}\left(a,b,q,p\right)$ are expressed as
\begin{align}\label{H_fun_abpq}
&~~~~H_{1}\left(\varrho,a,b,q,p\right)
\notag\\&
=\frac{\varrho^{p+1}\Gamma\left(q+1\right)}{\left(k-p\right)\left(\varrho b+a\right)^{q+1}}{}_2F_1\left(1,q+1;q-p+1;\frac{a}{\varrho b+a}\right)
\notag\\&
~~~~-\frac{\Gamma\left(q+1\right)}{\left(k-p\right)\left(b+a\right)^{q+1}}{}_2F_1\left(1,q+1;q-p+1;\frac{a}{b+a}\right)
\end{align}
and

\begin{align}
H_{2}\left(a,b,q,p\right)=&\frac{\Gamma\left(q+1\right)}{\left(k-p\right)\left(b+a\right)^{q+1}}
\notag\\
&\times{}_2F_1\left(1,q+1;q-p+1;\frac{a}{b+a}\right)
,\end{align}
respectively, in which  ${}_2F_1(\cdot,\cdot;\cdot;\cdot)$ denotes the hypergeometric function \cite[Eq. (9.100)]{Gradshteyn}.

\begin{proof}
See Appendix B.
\end{proof}

\end{theorem}
\begin{remark}\label{remark_L_SOP}
A large $L$ results in a small SOP over S-R link, since more antennas  bring a large
diversity gain at R.
\end{remark}

\section{Outage Analysis for R-D link}
As FSO is employed over R-D link, we assume that the information transmission from R to D will not be overheard by the eavesdroppers due to the highly directive and narrow nature of laser beam. Thus, in this section, we will investigate the impact of the randomness of the satellite's position on the outage performance of the information transmissions over R-D link.
\begin{figure}[!h]
\setlength{\abovecaptionskip}{0pt}
\setlength{\belowcaptionskip}{10pt}
\centering
\includegraphics[width=2.4 in]{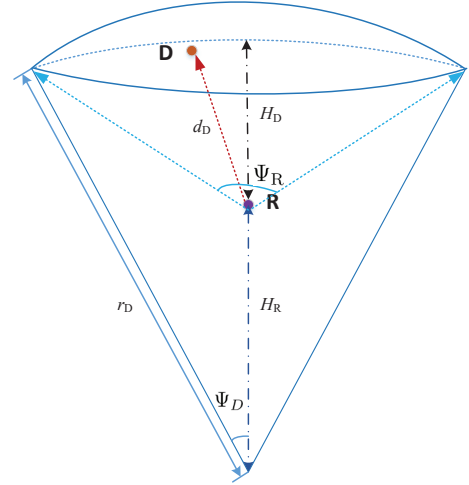}
\caption{The R-D link}
\label{RS}
\end{figure}

As shown in Fig. \ref{RS}, D is assumed to be uniformly distributed on the
spherical doom with the radius $r_\mathrm{D}=R_{\mathrm{earth}}+H_\mathrm{D}$,
where $R_{\mathrm{earth}}$ and $H_\mathrm{D}$ are the radius of Earth and
the orbit height of D, respectively.
The tracking
mechanisms in \cite{6809945,9672083} enable that R can localize D in real-time, which  can provide high directionality with FSO transmission for R-D link.
In order to facilitate the following analysis, spherical coordinates
are adopted, where the center of the earth is set as the original.
The relationship of $\Psi_{D}$ and the beamwidth of R  $\Psi_{\mathrm{R}}$ with low earth orbit S is
\begin{align}
\Psi_{\mathrm{R}}\approx2\tan^{-1}\left(\frac{R_{\mathrm{earth}}+H_{\mathrm{D}}}{H_{\mathrm{D}}}\sin\mathrm{\Psi_{\mathrm{D}}}\right).
\end{align}
Then, the coordinate
of D can be presented as $\left(r_\mathrm{D},\theta_\mathrm{D},\psi_\mathrm{D}\right)$,
where $0\le\theta_\mathrm{D}\le\mathrm{\Psi_\mathrm{D}}$  $\mathrm{\Psi_\mathrm{R}}$ and $0\le\psi_\mathrm{D}\le2\pi$.
The coordinate of R can be written as $(H_\mathrm{R},0,0)$, where $H_\mathrm{R}=R_{\mathrm{earth}}+h_\mathrm{R}$,
in which $h_{R}$ is the height of R. So the range of $H_R$ is
 $R_{\mathrm{earth}}+H_{\mathrm{min}} \text{\ensuremath{\le}}H_\mathrm{R} \text{\ensuremath{\le}}R_{\mathrm{earth}}+R_\mathrm{S}$.

Therefore, the distance between R and D, $d_\mathrm{D}$, can be presented
as
\begin{align}
d_\mathrm{D}=\sqrt{r_\mathrm{D}^{2}+H_\mathrm{R}^{2}-2r_\mathrm{D}H_\mathrm{R}\cos\theta_\mathrm{D}}.
\end{align}
Since $r_\mathrm{D}-H_\mathrm{R}=d_{\mathrm D,\mathrm{min}}\text{\ensuremath{\le}}d_\mathrm{D}\text{\ensuremath{\le}}d_{\mathrm D,\mathrm{max}}=\sqrt{r_\mathrm{D}^{2}+H_\mathrm{R}^{2}-2r_\mathrm{D}H_\mathrm{R}\cos\Psi_\mathrm{D}}$, we consider the case that the path-loss factor is 2 to simplify the analysis in the following.
\begin{prop}\label{dSR}
When D is uniformly distributed in the space shown in Fig. \ref{RS} and the height of R is known,  the PDF of $d_{D}^2$ can derived as
\begin{align}\label{pdf_dSR2}
f_{d_\mathrm{D}^{2}}\left(w\right)=\frac{1}{2r_\mathrm{D}H_{R}\left(1-\cos\Psi_\mathrm{D}\right)}
,
\end{align}
and the range of $d_{D}^{2}$ is $\left(r_\mathrm{D}-H_\mathrm{R}\right)^{2}=d_{\mathrm D,\mathrm{min}}^{2}\text{\ensuremath{\le}}d_\mathrm{D}^{2}\text{\ensuremath{\le}}d_{\mathrm D,\mathrm{max}}^{2}=r_\mathrm{D}^{2}+H_\mathrm{R}^{2}-2r_\mathrm{D}H_\mathrm{R}\cos\Psi_\mathrm{D}$.

\begin{proof}
See Appendix C.
\end{proof}

\end{prop}

The OP of the link between R and D is defined as
\begin{align}\label{op2_define}
{\rm{OP}_2} = \Pr \left\{  \gamma_\mathrm{D} < \gamma_{\mathrm{out}}\right\},
\end{align}
where $\gamma_{\mathrm{out}}$ is the threshold that enables D to effectively receive the signals from R.

\begin{theorem}\label{OP2}
When D is uniformly distributed in the space shown in Fig. \ref{RS} and the height of R is known,  the OP for R-D link of the considered  RF-FSO cooperative SATN  is obtained as
\begin{align}\label{op2_final}
\widetilde{\mathrm{OP_{2}}}=&\frac{I}{2 \epsilon r_\mathrm{D}H_\mathrm{R}\left(1-\cos\Psi_\mathrm{D}\right)}
\notag\\&\times
\left\{G_{r+2,3r+2}^{3r,2}\left[\epsilon d_{\mathrm D,\mathrm{max}}^{2}\Big| \begin{array}{c}
1,2,K_{1}+1\\
K_{2}+1,0,1
\end{array}\right]\right.
\notag\\&~~~~
\left.-G_{r+2,3r+2}^{3r,2}\left[\epsilon d_{\mathrm D,\mathrm{min}}^{2}\Big|  \begin{array}{c}
1,2,K_{1}+1\\
K_{2}+1,0,1
\end{array}\right]\right\}
,\end{align}
where
\begin{align}
\epsilon=\frac{\left({4\pi f_\mathrm c}\sigma_{d}\right)^{2}\left(hab\right)^{r}}{P_\mathrm{R}{c}^2\zeta^{2}\mathcal{L}_\mathrm{r}^{2}r^{2r}}\gamma_{\mathrm{out}}
.\end{align}
\begin{proof}
See Appendix D.
\end{proof}

\end{theorem}

%

Based on \emph{Theorem \ref{OP2}} and the randomness of the positions of R  in Fig. \ref{SR} and D  in Fig. \ref{RS}, the OP for R-D link of the RF-FSO cooperative SATN  is obtained.

\begin{theorem}\label{OP2_final}
Considering the randomness of the positions of R  in Fig. \ref{SR} and D  in Fig. \ref{RS}, the OP for R-D link of the RF-FSO cooperative SATN  is obtained as
\begin{align}\label{op2_final_app}
{\mathrm{OP}}_{2}=&\frac{I}{2 \epsilon r_\mathrm{D}\bar H_\mathrm{R}\left(1-\cos\Psi_\mathrm{D}\right)}
\notag\\&\times
\left\{G_{r+2,3r+2}^{3r,2}\left[\epsilon \bar d_{\mathrm D,\mathrm{max}}^{2}\Big| \begin{array}{c}
1,2,K_{1}+1\\
K_{2}+1,0,1
\end{array}\right]\right.
\notag\\&~~~~
\left.-G_{r+2,3r+2}^{3r,2}\left[\epsilon {\bar d_{\mathrm D,\mathrm{min}}^{2}}\Big|\begin{array}{c}
1,2,K_{1}+1\\
K_{2}+1,0,1
\end{array}\right]\right\}
,\end{align}
where $\bar H_\mathrm{R}$, $\bar d_{\mathrm D,\mathrm{max}}^{2}$, and $\bar d_{\mathrm D,\mathrm{min}}^{2}$ are  the medians of   $ H_\mathrm{R}$, $d_{\mathrm D,\mathrm{max}}^{2}$, and $d_{\mathrm D,\mathrm{min}}^{2}$.
\begin{proof}
See Appendix E.
\end{proof}
\end{theorem}
\begin{remark}\label{remark_op2_P}
The OP over R-D link   approaches zero when the transmit power at R goes infinity.
\end{remark}
\begin{coro}\label{coro2}
The randomness of R have some effects on  $\mathrm{SOP}_{1}$ but have little effect on $\mathrm{OP_{2}}$, so $\mathrm{SOP}_{1}$ and $\mathrm{OP_{2}}$  can be viewed as independent variable.
\end{coro}

\begin{proof}
See Appendix F.
\end{proof}

\section{SOP of the RF-FSO cooperative SATN}
\subsection{Content Popularity Model }
In this section, we define the content popularity distribution model.
Note that content is delivered in form of  files in this paper.
The file in caches of R depends on its popularity. $\mathcal{N}$
indicates that collection of files with total files number $N=|\mathcal{N}|$.
 All files are arranged in the descending order of popularity, where
more popular files are ranked with smaller indices. We assume that
the probability of file requests follows the Zipf\textquoteright s
law, as given by \cite{MZIP}

\begin{equation}
p_{n}=\zeta_\mathrm{zip} n^{-\alpha},
\end{equation}
where $\zeta_\mathrm{zip}=\left(\sum_{n=1}^{N}\frac{1}{n^{\alpha}}\right)^{-1}$ is the normalization
factor to have $\sum_{n=1}^{N}p_{n}=1$ and the parameter $\alpha$ is
the skewness parameter which controls the distribution tail.

With the limited storage capacity, R adopts the MPC caching scheme. It means only the most popular $M$ files are stored at R, where  $M$ is the R's storage normalized by the size of each file.
\subsection{SOP of the RF-FSO cooperative SATN}

Based on the MPC caching scheme at R, and the above analysis of SOP/OP in two phases of the considered system in \emph{Theorem \ref{SOP1_theo}},  \emph{Theorem \ref{OP2_final}} and \emph{Corollary \ref{coro2}},   the lower bound of
the end-to-end
SOP over the uplink of the  RF-FSO cooperative SATN   is
\begin{align}\label{sop}
\mathrm{SOP}=\varphi\mathrm{OP_{2}}+\varphi^{'}\left[1-\left(1-\mathrm{SOP}_{1}^{\mathrm{L}}\right)\left(1-{\mathrm{OP}}_{2}\right)\right],
\end{align}
where $\varphi=\sum_{n=1}^{M}p_{n}$,  $\varphi^{'}=\sum_{n=M+1}^{N}p_{n}=1-\varphi$, $\mathrm{SOP}_{1}^{\mathrm{L}}$ and ${\mathrm{OP}}_{2}$ are  available in \eqref{sop1_final} and \eqref{op2_final_app}, respectively.

\begin{remark}\label{remark_M_alpha}
The larger the storage capacity of R (the larger the skewness parameter) is, the larger the value of $\varphi$, the stronger  dependencies of  OP over R-D link to SOP of the considered SATN are achieved.  On the contrary, the SOP of the considered SATN is more dependent on the SOP over S-R link.
\end{remark}
\begin{remark}\label{remark_OP_equal_SOP}
When  the requested file is cached at R, the   end-to-end  SOP of the RF-FSO cooperative SATN is the  OP over the R-D link, which disregards the S-R link.
\end{remark}
\begin{remark}\label{remark_MPC_SOP}
 The adopted MPC caching scheme can  bring the improved secrecy outage performance with a higher skewness
parameter.
\end{remark}

\section{Numerical Results and Discussion}
In this section,  Monte-Carlo simulation results will be presented to study
the performance of the considered system and demonstrate the correctness of the analytical results.
In the following, we run $1\times10^{6}$ trials of Monte-Carlo simulations, to model
the randomness of the positions of the considered Eves and UAV and channel gains over each link.
Besides, the transmit power at S is fixed \cite{Sharma_uav_cach}, the caching memory $M$ is based on the storage of R and the  size of the requested file model \cite{6763007}, $\Psi_\mathrm{D}=\frac{\pi}{12}$ is related to the beamwidth and the height of  R \cite{9684552} and the optical-to-electrical conversion coefficient is set as $\zeta= 0.5$ \cite{fso_DL_TWC}.
The main parameters are set as
$H_\mathrm{D}=550$ km, $R_\mathrm{S}=300$ m,
$R_{\mathrm{earth}}=6371$ km,
 $\mathcal{L}_\mathrm{r}=81$ dB,  $N_{R}=1$ W, $N_\mathrm{E}=1$ W, $m_\mathrm{R}=2$, $\Omega_\mathrm{R}=1.9$, $m_\mathrm{E}=2$, $\Omega_\mathrm{E}=0.5$, $P_\mathrm{S}=10$ dBW, and $C_\mathrm{th}=0.01$ bits/s.

Figs. 4-8 show the SOP versus $P_\mathrm{R}$ for different values of $a$, $b$, $r$, $M$, $\zeta$, $L$, $\alpha$, $H_\mathrm{D}$, $K$,  and different
caching schemes. These figures  show that the secrecy
performance improves with increasing $P_\mathrm{R}$, as the SNR at the satellite is improved.
Moreover,  the SOP becomes saturated when the transmit power at R gets higher. That is, the SOP exhibits a floor because the secrecy capacity will
become a constant, as reported in \cite{Lei_2017}. Here, we mainly discuss the influence of the transmit power at R on the SOP, while the  transmit power at S is fixed.
When the  transmit power at R increases enough, the OP of the R-D link  approaches zero, and the SOP of the considered SATN is the SOP of the S-R link, which is a fixed value in presence of the predetermined simulation conditions.
Since the transmit power at R is  constrained, the  simulation results can provide  some hints to researchers to design systems while considering both  reliability and power consumption, and a trade-off can be arrived to get relatively  high reliability and acceptable power consumption  with exiting constrains.

\begin{figure}[!h]
\setlength{\abovecaptionskip}{0pt}
\setlength{\belowcaptionskip}{10pt}
\centering
\includegraphics[height= 2.7 in, width=3.8 in]{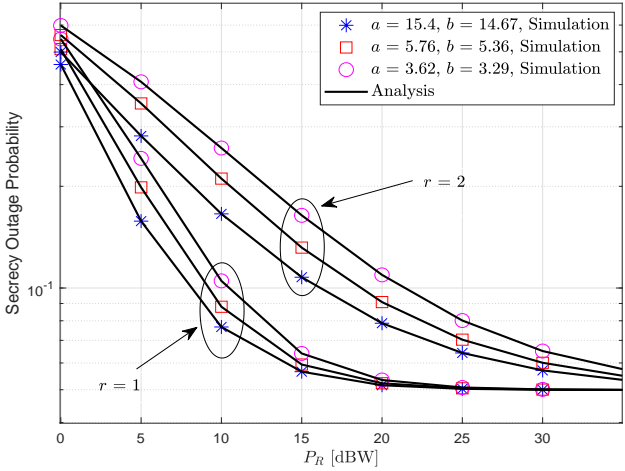}
\caption{SOP versus $P_R$ for
$H_\mathrm{D}=550$ km, $R_\mathrm{S}=300$ m, $H_{\mathrm{min}}=80$ m, $\Psi_\mathrm{D}=\frac{\pi}{12}$,
$R_{\mathrm{earth}}=6371$ km,  $\xi=1.1$,
  $\mathcal{L}_\mathrm{r}=81$ dB, $M=10$, $\alpha=1$, $N=10^{6}$, $N_\mathrm{R}=1$ W, $N_\mathrm{E}=1$
W, $K=3$, $L=8$, $m_\mathrm{R}=2$, $\Omega_\mathrm{R}=1.9$, $m_\mathrm{E}=2$, $\Omega_\mathrm{E}=0.5$, $P_\mathrm{S}=10$ dBW, and $C_\mathrm{th}=0.01$ bits/s.}
\label{FSO_ab}
\end{figure}
Fig. \ref{FSO_ab} represents the SOP for different values $(a, b)$ and
$r$.  One can also see that the SOP
with the weakest turbulence $(a = 15.4, b = 14.67)$ is lower
than that with strongest turbulence $(a = 3.62, b = 3.29)$.
 We know $r$ represents the detection scheme used at D, where $r = 1$ is for HD and $r = 2$ is for IM/DD.
By varying $r$ and keeping $(a, b)$ fixed
in Fig. \ref{FSO_ab}, the HD detection method can lead to better secrecy
performance than IM/DD method.  The reason for this
is that the SNR obtained with the HD method is higher than
that of IM/DD.
Fig. \ref{FSO_ab} implies that the weaker turbulence on FSO link and the HD detection method used at D
can enhance the secure performance of the considered SATN.

\begin{figure}[!h]
\setlength{\abovecaptionskip}{0pt}
\setlength{\belowcaptionskip}{10pt}
\centering
\includegraphics[height= 2.7 in, width=3.8 in]{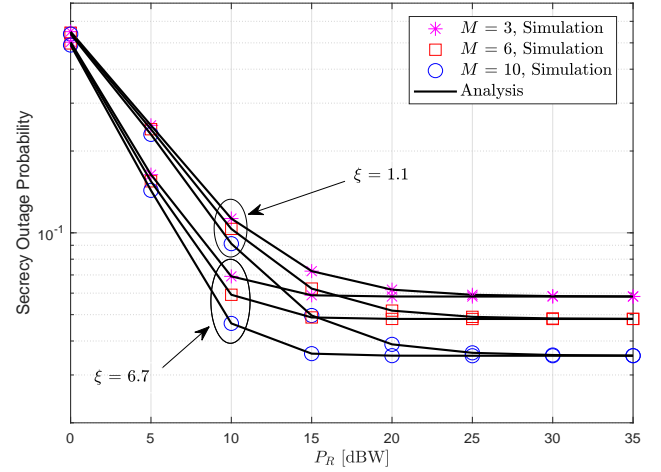}
\caption{SOP versus $P_R$ for
$H_\mathrm{D}=550$ km, $R_\mathrm{S}=300$ m, $H_{\mathrm{min}}=80$ m, $\Psi_\mathrm{D}=\frac{\pi}{12}$,
$R_{\mathrm{earth}}=6371$ km, $a=15.47$, $b=14.6$,
$r=2$,  $\mathcal{L}_\mathrm{r}=81$ dB, $\alpha=1$, $N=10^{6}$, $N_\mathrm{R}=1$ W, $N_\mathrm{E}=1$
W, $K=3$, $L=8$, $m_\mathrm{R}=2$, $\Omega_\mathrm{R}=1.9$, $m_\mathrm{E}=2$, $\Omega_\mathrm{E}=0.5$, $P_\mathrm{S}=10$ dBW, and $C_\mathrm{th}=0.01$ bits/s.}
\label{xi_M}
\end{figure}
In Fig. \ref{xi_M}, the size of the caching memory shows a positive influence on the SOP of
the considered system, since a large $M$  means a large probability of omitting the terrestrial terminal-relay link, which is  consistent with \emph{Remark \ref{remark_M_alpha}}.
The MPC caching scheme achieves the largest cooperative diversity gain with largest $M$ which improve the secure performance.
Moreover, the SOP with lower $\xi$ is higher than that with larger $\xi$, because a larger $\xi$ means high pointing accuracy over R-D link.
Simulation results can provide  some hints to researchers to upgrade SOP in terms of pointing accuracy and caching memory of R with power consumption constraints.

\begin{figure}[!h]
\setlength{\abovecaptionskip}{0pt}
\setlength{\belowcaptionskip}{10pt}
\centering
\includegraphics[height= 2.7 in, width=3.8 in]{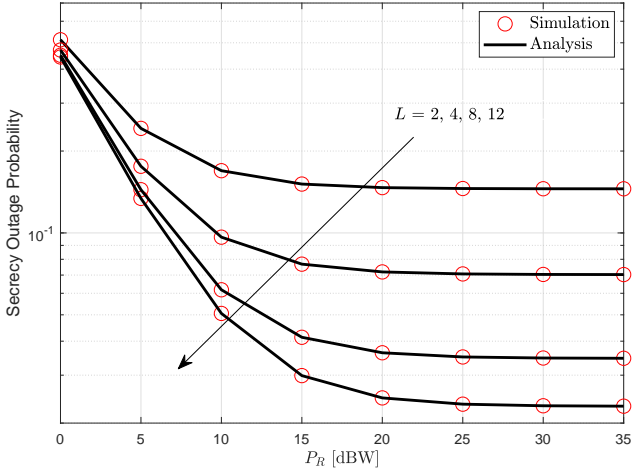}
\caption{SOP versus $P_R$ for
$H_\mathrm{D}=550$ km, $R_\mathrm{S}=300$ m, $H_{\mathrm{min}}=80$ m, $\Psi_\mathrm{D}=\frac{\pi}{12}$,
$R_{\mathrm{earth}}=6371$ km, $a=15.47$, $b=14.6$, $\xi=1.1$,
$r=2$,  $\mathcal{L}_\mathrm{r}=81$ dB, $M=10$, $\alpha=1$, $N=10^{6}$, $N_\mathrm{R}=1$ W, $N_\mathrm{E}=1$
W, $K=3$, $m_\mathrm{R}=2$, $\Omega_\mathrm{R}=1.9$, $m_\mathrm{E}=2$, $\Omega_\mathrm{E}=0.5$, $P_\mathrm{S}=10$ dBW, and $C_\mathrm{th}=0.01$ bits/s.}
\label{Pr_L}
\end{figure}

The influence of the antenna number at R on the SOP
performance is depicted in Fig. \ref{Pr_L}. As expected,
$L$ exhibits a positive effect.
A large $L$  results in a small
SOP, meaning better secrecy performance, since a large $L$
can bring a large diversity gain at R, as described in \emph{Remark \ref{remark_L_SOP}}.
Moreover, the SOP with a lower $\xi$ is higher than that with a larger $\xi$, because a larger $\xi$ means high pointing accuracy over R-D link.
Simulation results can offer  some hints to researchers to upgrade SOP in terms of pointing accuracy and caching memory of R with power consumption constraints.

In Fig. \ref{Pr_Hmin_alpha}, the SOP performance is studied
for various $H_{\mathrm{min}}$ and $\alpha$.
It is obvious that a small $H_{\mathrm{min}}$ leads to a small SOP.
This observation comes from the fact that a small $H_{\mathrm{min}}$ denotes
a low  path-loss average for the signal transmitted over S-R link, while the path-loss average for the signal transmitted over S-D link is fixed with the same coverage space.
Moreover, a high $\alpha$ leads to improved SOP. This is because a low $\alpha$ means that more files have similar popularity and a high $\alpha$ means that a few files have a very high popularity and a large number of files have very low popularity. In other
words, $\alpha$ indicates a high or low degree of request concentration.
This implies that a large $\alpha$ and the adopted MPC caching scheme can jointly bring improved secrecy outage performance, which is  illustrated in  \emph{Remark \ref{remark_M_alpha}} and  \emph{\ref{remark_MPC_SOP}}.

To provide a fair performance comparison in terms of the secrecy performance of the considered SATN, Fig. \ref{Pr_K} describes the SOP employing MPC caching, UC caching, and no caching. Obviously, the MPC caching  achieves the
largest cooperative diversity gain  and is the best in terms of SOP with the file request following the Zipf\textquoteright s law.
Moreover, Fig. \ref{Pr_K} presents the secrecy outage performance  for
various $K$, while $P_R$ increasing. Obviously, $K$ shows a negative
impact on the secrecy outage performance, as a large $K$ means that Eves are distributed around terrestrial terminals more densely. This explains that SOP degrades when the number of Eves increases.

\begin{figure}[!h]
\setlength{\abovecaptionskip}{0pt}
\setlength{\belowcaptionskip}{10pt}
\centering
\includegraphics[height= 2.7 in, width=3.8 in]{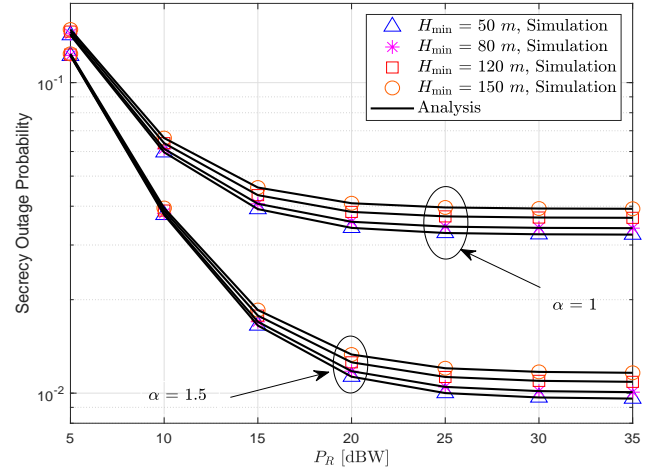}
\caption{SOP versus $P_R$ for
$H_\mathrm{D}=550$ km, $R_\mathrm{S}=300$ m,  $\Psi_\mathrm{D}=\frac{\pi}{12}$,
$R_{\mathrm{earth}}=6371$ km, $a=15.47$, $b=14.6$, $\xi=1.1$,
$r=2$,  $\mathcal{L}_\mathrm{r}=81$ dB, $M=10$, $\alpha=1$, $N=10^{6}$, $N_\mathrm{R}=1$ W, $N_\mathrm{E}=1$
W, $K=3$, $L=8$, $m_\mathrm{R}=2$, $\Omega_\mathrm{R}=1.9$, $m_\mathrm{E}=2$, $\Omega_\mathrm{E}=0.5$, $P_\mathrm{S}=10$ dBW, and $C_\mathrm{th}=0.01$ bits/s.}
\label{Pr_Hmin_alpha}
\end{figure}

\begin{figure}[!h]
\setlength{\abovecaptionskip}{0pt}
\setlength{\belowcaptionskip}{10pt}
\centering
\includegraphics[height= 2.7 in, width=3.8 in]{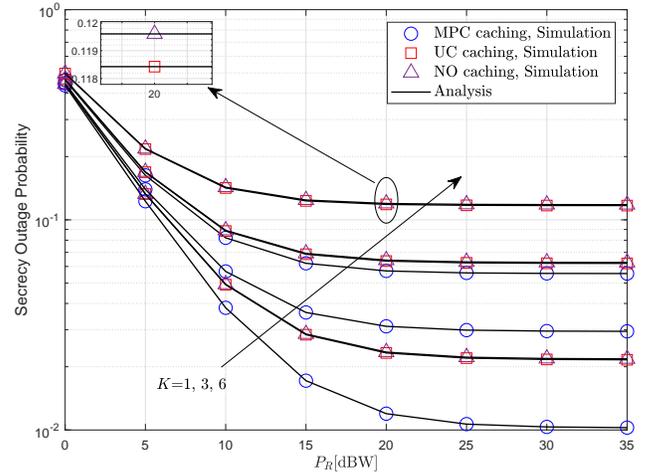}
\caption{
SOP versus $P_R$ for
$H_\mathrm{D}=550$ km, $R_\mathrm{S}=300$ m, $H_{\mathrm{min}}=80$ m, $\Psi_\mathrm{D}=\frac{\pi}{12}$,
$R_{\mathrm{earth}}=6371$ km, $a=15.47$, $b=14.6$, $\xi=1.1$,
$r=2$,  $\mathcal{L}_\mathrm{r}=81$ dB, $M=100$, $\alpha=1$, $N=10^{4}$, $N_\mathrm{R}=1$ W, $N_\mathrm{E}=1$
W,  $L=8$, $m_\mathrm{R}=2$, $\Omega_\mathrm{R}=1.9$, $m_\mathrm{E}=2$, $\Omega_\mathrm{E}=0.5$, $P_\mathrm{S}=10$ dBW, and $C_\mathrm{th}=0.01$ bits/s.}
\label{Pr_K}
\end{figure}

\section{Conclusions}
In this paper, we have investigated the SOP of the uplink transmission of an RF-FSO hybrid cooperative SATN in the presence of a group of aerial Eves. In the considered cooperative SATN, S transmits its information to D via the help of a cache-enabled R, while a group of Eves try to overhear the transmitted confidential information. Considering the randomness of the positions of R, D, and Eves, and employing stochastic geometry, the secrecy outage performance of the cooperative uplink transmission in the considered SATN has been investigated and the closed-form analytical expression for the end-to-end SOP has been derived. Numerical results have demonstrated the accuracy of the derived expressions and presented some trends of SOP by altering some system parameters of interest, which  can facilitate
researchers to design systems considering both reliability and
exiting constrains. For example, the transmit power
at R, the number of antennas at R, the  minimum height of R, the detection
scheme used at D, etc, while a trade-off can be arrived to get
relatively high reliability with exiting constrains.
The ergodic secrecy capacity  of the considered SATN  are expected to be investigated in future work.

\ifCLASSOPTIONcaptionsoff
  \newpage
\fi

\appendices

\section{}

If $Z=\frac{d_\mathrm{R}^{\eta_1}}{d_\mathrm{E}^{\eta_1}}\leq1$, the CDF of $Z$ is
\begin{align}
F_{Z}\left(z\right)=&\intop_{H_{\mathrm{min}}^{\eta_1}/z}^{R_\mathrm{S}^{\eta_1}}\intop_{H_{\mathrm{min}}^
{\eta_1}/y}^{z}\frac{2\pi y}{\eta_1 V_{ S_{1}}}\left(\left(yu\right)^{\frac{3}{\eta_1}-1}-H_{\mathrm{min}}\left(yu\right)^{\frac{2}
{\eta_1}-1}\right)
\notag\\
&
\times\sum_{f=1}^{K}{K \choose f}\left(-1\right)^{f+1}\frac{3f}{\eta_1 R_\mathrm{S}^{3f}}y^{\frac{3f}{\eta_1}-1}dudy
.\end{align}

After some mathematical manipulations by using polynomial integration,  $F_{Z}\left(z\right)$ is expressed as \eqref{F_z_large} shown on the top of next page.
\begin{figure*}
\begin{align}\label{F_z_large}
F_{Z}\left(z\right)=&\frac{2\pi}{V_{ S_{1}}}\sum_{f=1}^{K}{K \choose f}\left(-1\right)^{f+1}\frac{3f}{R_\mathrm{S}^{3f}}
\left[
\frac{1}{9f+9}\left(R_\mathrm{S}^{3f+3}z^{\frac{3}{\eta_1}}-H_{\mathrm{min}}^{3f+3}z^{-\frac{3f}{\eta_1}}
\right)-\frac{1}{9f}H_{\mathrm{min}}^{3}\left(R_\mathrm{S}^{3f}-H_{\mathrm{min}}^{3f}z^{-\frac{3f}{\eta_1}}
\right)
\right.
\notag\\
&
\left.
-\frac{1}{2\left(3f+2\right)}H_{\mathrm{min}}\left(R_\mathrm{S}^{3f+2}z^{\frac{2}{\eta_1}}-H_{\mathrm{min}}
^{3f+2}z^{-\frac{3f}{\eta_1}}\right)+\frac{1}{6f}H_{\mathrm{min}}^{3}\left(R_\mathrm{S}^{3f}-H_{\mathrm{min}}
^{3f}z^{-\frac{3f}{\eta_1}}\right)
\right]
\end{align}
\rule{18cm}{0.01cm}
\end{figure*}

Using the differentiation, the PDF of $Z$ is derived as
\begin{align}
f_{Z}\left(z\right)=\sum_{f=1}^{K}\left(B_{1,f}z^{\frac{3}{\eta_1}-1}+B_{2,f}z^{\frac{2}{\eta_1}-1}+B_{3,f}z^{-\frac{3f}{\eta_1}-1}\right)
,\end{align}
where $B_{1,f}$, $B_{2,f}$ and $B_{3,f}$ are expressed in \eqref{B1f}, \eqref{B2f}, and \eqref{B3f}, respectively.

Similarly, if $Z=\frac{d_\mathrm{R}^{\eta_1}}{d_\mathrm{E}^{\eta_1}}>1$, the CDF of $Z$ is
\begin{align}
F_{Z}\left(z\right)=F_{1}\left(z\right)+F_{2}\left(z\right)
,\end{align}
where
\begin{align}
F_{1}\left(z\right)=&\intop_{R_\mathrm{S}^{\eta_1}/z}^{R_{S}^{\eta_1}}\intop_{H_{\mathrm{min}}^{\eta_1}/y}^{R_\mathrm{S}^{\eta_1}
/y}y\frac{2\pi}{\eta_1 V_{\mathrm S_{1}}}\left(\left(yu\right)^{\frac{3}{\eta_1}-1}-H_{\mathrm{min}}\left(yu\right)^{\frac{2}{\eta_1}-1}\right)
\notag\\
&
\times\sum_{f=1}^{K}{K \choose f}\left(-1\right)^{f+1}\frac{3f}{\eta_1 R_\mathrm{S}^{3f}}y^{\frac{3f}{\eta_1}-1}dudy
\end{align}
and
\begin{align}
F_{2}\left(z\right)=&\intop_{H_{\mathrm{min}}^{\eta_1}/z}^{R_\mathrm{S}^{\eta_1}/z}\intop_{H_{\mathrm{min}}^{\eta_1}/y}^{z}y
\frac{2\pi}{\eta_1 V_{S_{1}}}\left(\left(yu\right)^{\frac{3}{\eta_1}-1}-H_{\mathrm{min}}\left(yu\right)^{\frac{2}{\eta_1}-1}\right)
\notag\\
&
\times\sum_{f=1}^{K}{K \choose f}\left(-1\right)^{f+1}\frac{3f}{\eta_1 R_\mathrm{S}^{3f}}y^{\frac{3f}{\eta_1}-1}dudy
.\end{align}

After some mathematical manipulations, $F_{1}\left(z\right)$ and $F_{2}\left(z\right)$ are obtained as \eqref{F1z} and \eqref{F2z} shown on the top of next page.

\begin{figure*}
\begin{align}\label{F1z}
F_{1}\left(z\right)=\frac{2\pi}{\eta_1 V_{ S_{1}}}\sum_{f=1}^{K}{K \choose f}\left(-1\right)^{f+1}\frac{3f}{ R_\mathrm{S}^{3f}}\left[\frac{\eta_1}{3f}R_\mathrm{S}^{3f}\left(1-z^{-\frac{3f}{\eta_1}}\right)\right]
\left(\frac{1}{3}R_\mathrm{S}^{3}-\frac{1}{2}H_{\mathrm{min}}R_\mathrm{S}^{2}+\frac{1}{6}H_{\mathrm{min}}^{3}
\right)
\end{align}
\rule{18cm}{0.01cm}
\end{figure*}

\begin{figure*}
\begin{align}\label{F2z}
F_{2}\left(z\right)=&\frac{2\pi}{V_{ S_{1}}}\sum_{f=1}^{K}{K \choose f}\left(-1\right)^{f+1}\frac{3f}{R_\mathrm{S}^{3f}}z^{-\frac{3f}{\eta_1}}\left(\frac{1}{3\left(3f+3\right)}R_\mathrm{S}
^{3f+3}-\frac{1}{2
\left(3f+2\right)}R_\mathrm{S}^{3f+2}H_{\mathrm{min}}\right)
\notag\\
&
+\left(\frac{1}{2\left(3f+2\right)}+\frac{1}{9f}-\frac{1}{3\left(3f+3\right)}-\frac{1}{6f}\right)H_{\mathrm{min}}^{3f+3}+
\left(\frac{1}{6f}-\frac{1}{9f}\right)R_\mathrm{S}^{3f}H_{\mathrm{min}}^{3}
\end{align}
\rule{18cm}{0.01cm}
\end{figure*}

Using the differentiation, the PDF of $Z$ is derived as
\begin{align}
f\left(z\right)=\sum_{f=1}^{K}A_{f}z^{-\frac{3f}{\eta_1}-1}
,
\end{align}
where
\begin{align}
A_{f}=A_{1,f}+A_{2,f}
.
\end{align}
in which $A_{1,f}$ and $A_{2,f}$ are expressed in \eqref{A1f} and \eqref{A2f}, respectively.

Thus, we can derive the PDF of $Z$ as \eqref{dZ}.
\section{}
According to \eqref{sop_1}, $\mathrm{SOP}_{1}$ can be presented as
\begin{align}
\mathrm{SOP}_{1}&=\Pr\left\{ \frac{P_\mathrm{S}\left\Vert \mathbf{h}_\mathrm{SR}\right\Vert ^{2}}{N_\mathrm{R}d_\mathrm{R}^{\eta_1}}\leq\lambda\frac{P_\mathrm{S}\left\Vert \mathbf{h}_{\mathrm {SE}^{*}}\right\Vert ^{2}}{N_\mathrm{E}d_\mathrm{E}^{\eta_1}}+\lambda-1\right\}
.
\end{align}

We know $\lambda=2^{C_\mathrm{th}}>1$ because $C_\mathrm{th}>0$. Thus, we have
\begin{align}
\mathrm{SOP}_{1}&\geq\Pr\left\{ \frac{P_\mathrm{S}\left\Vert \mathbf{h}_\mathrm{SR}\right\Vert ^{2}}{N_\mathrm{R}d_\mathrm{R}^{\eta_1}}\leq\lambda\frac{P_\mathrm{S}\left\Vert \mathbf{h}_{\mathrm{SE}^{*}}\right\Vert ^{2}}{N_\mathrm{E}d_\mathrm{E}^{\eta_1}}\right\}
\notag\\
&
=\Pr\left\{ \left\Vert \mathbf{h}_\mathrm{SR}\right\Vert ^{2}\leq N_\mathrm{R}d_\mathrm{R}^{\eta_1}\lambda\frac{\left\Vert \mathbf{h}_{\mathrm{SE}^{*}}\right\Vert ^{2}}{N_\mathrm{E}d_\mathrm{E}^{\eta_1}}\right\} =\mathrm{SOP}_{1}^{\mathrm{L}}
.
\end{align}

Let $X=\left\Vert \mathbf{h}_{\mathrm{SE}^{*}}\right\Vert ^{2}$, $Z=\frac{d_\mathrm{R}^{\eta_1}}{d_\mathrm{E}^{\eta_1}}$.   $\mathrm{SOP}_{1}^{\mathrm{L}}$ can be obtained as
\begin{align}
&\mathrm{SOP}_{1}^{\mathrm{L}}=\Pr\left\{ \left\Vert \mathbf{h}_\mathrm{SR}\right\Vert ^{2}\leq\frac{N_\mathrm{R}}{N_\mathrm{E}}\lambda XZ\right\}
\notag\\&
=\Pr\left\{ \left\Vert \mathbf{h}_\mathrm{SR}\right\Vert ^{2}\leq\frac{N_\mathrm{R}}{N_\mathrm{E}}\lambda XZ\right\}
\notag\\&
=\intop_{\frac{H_{\mathrm{min}}^{\eta_1}}{R_\mathrm{S}^{\eta_1}}}^{\infty}\intop_{0}^{\infty}F_{\left\Vert \mathbf{h}_\mathrm{SR}\right\Vert ^{2}}\left(\frac{N_\mathrm{R}}{N_\mathrm{E}}\lambda xz\right)f_{X}\left(x\right)dxf_{Z}\left(z\right)dz
.\end{align}

Substituting the CDF of $\left\Vert \mathbf{h}_\mathrm{SR}\right\Vert ^{2}$, $\mathrm{SOP}_{1}^{\mathrm{L}}$ is expressed as
\begin{align}
\mathrm{SOP}_{1}^{\mathrm{L}}=&1-\sum_{k=0}^{Lm_\mathrm{R}-1}\frac{\left(a_{0}\right)^{k}}{k!}\intop_{\rho_Z}
^{\infty}\intop_{0}^{\infty}\exp\left(-a_{0}xz\right)
\notag\\&
\times x^{k}z^{k}f_{X}\left(x\right)dxf_{Z}\left(z\right)dz
,\end{align}
where $a_{0}=\lambda \lambda_{R}\frac{N_\mathrm{R}}{N_\mathrm{E}}$ and $\rho_Z=\frac{H_{\mathrm{min}}^{\eta_1}}{R_\mathrm{S}^{\eta_1}}$.

Substituting the PDF of $\left\Vert \mathbf{h}_{\mathrm{SE}^{*}}\right\Vert ^{2}$ and $\frac{d_\mathrm{R}^{\eta_1}}{d_\mathrm{E}^{\eta_1}}$,  $\mathrm{SOP}_{1}^{\mathrm{L}}$ is obtained as \eqref{sop1_pdf} shown on the top of next page.
\begin{figure*}
\begin{align}\label{sop1_pdf}
&\mathrm{SOP}_{1}^{\mathrm{L}}=1-\frac{\lambda_\mathrm{E}^{m_\mathrm{E}}}{\Gamma\left(m_\mathrm{E}\right)}\sum_{k=0}^
{Lm_{R}-1}\frac{a_{0}^{k}}{k!}\intop_{0}^{\infty}x^{k+m_{E}-1}\exp\left(-\lambda_{E}x\right)
\notag\\&~~~~
\times\left[\sum_{f=1}^{K}\intop_{\rho_Z}^{1}\exp\left(-a_{0}xz\right)\left(B_{1,f}z^{k+\frac{3}
{\eta_1}-1}+B_{2,f}z^{k+\frac{2}{\eta_1}-1}+B_{3,f}z^{k-\frac{3f}{\eta_1}-1}\right)dz+\sum_{f=1}^{K}
A_{f}\intop_{1}^{\infty}\exp\left(-a_{0}xz\right)z^{k-\frac{3f}{\eta_1}-1}dz\right]dx
\end{align}
\rule{18cm}{0.01cm}
\end{figure*}

Changing the range of integration, $\mathrm{SOP}_{1}^{\mathrm{L}}$ is shown in \eqref{six_group_sop1} on the top of next page, where

\begin{figure*}
\begin{align}\label{six_group_sop1}
\mathrm{SOP}_{1}^{\mathrm{L}}=&1-\frac{\lambda_\mathrm{E}^{m_\mathrm{E}}}{\Gamma\left(m_\mathrm{E}\right)}\sum_{k=0}^{Lm_\mathrm{R}-1}\sum_{f=1}^{K}\frac{a_{0}^{k}}{k!}
\left\{
B_{1,f}\left[H\left(\rho,\lambda_\mathrm{E},a_{0},k+m_\mathrm{E}-1,k+\frac{3}{\eta_1}-1\right)-H\left(1,\lambda_\mathrm{E},a_{0},k+m_\mathrm{E}-1,k+\frac{3}{\eta_1}-1\right)\right]
\right.\notag\\&
+B_{2,f}\left[H\left(\rho,\lambda_\mathrm{E},a_{0},k+m_\mathrm{E}-1,k+\frac{2}{\eta_1}-1\right)-H\left(1,\lambda_\mathrm{E},a_{0},k+m_\mathrm{E}-1,k+\frac{2}{\eta_1}-1\right)\right]
\notag\\&\left.
+B_{3,f}\left[H\left(\rho,\lambda_\mathrm{E},a_{0},k+m_\mathrm{E}-1,k-\frac{3f}{\eta_1}-1\right)+(A_{f}-B_{3,f})H\left(1,\lambda_\mathrm{\mathrm E},a_{0},k+m_{E}-1,k-\frac{3f}{\eta_1}-1\right)\right]
\right\}
\end{align}
\rule{18cm}{0.01cm}
\end{figure*}
\begin{align}
H\left(\varrho,a,b,q,p\right)=\intop_{0}^{\infty}\intop_{\rho_Z}^{\infty}\exp\left(-ax\right)x^{q}\exp\left(-bxz\right)z^{p}dxdz
.\end{align}
Next, we will analyze the function $H\left(\varrho,a,b,q,p\right)$. When $p\geq0$, using \cite[Eq. (3.381.3)]{Gradshteyn}, $H\left(\varrho,a,b,q,p\right)$ is expressed as
\begin{align}
&H\left(\varrho,a,b,q,p\right)=\intop_{0}^{\infty}x^{q}\exp\left(-ax\right)\left[\intop_{\rho_Z}^{\infty}\exp\left(-bxz\right)z^{p}dz\right]dx
\notag\\&
~~~~~~~~~~~~~~~~=\frac{1}{b^{p+1}}\intop_{0}^{\infty}x^{q-p-1}\exp\left(-ax\right)\Gamma\left(p+1,b\rho_Z x\right)dx
.\end{align}

Then, using \cite[Eq. (6.455.1)]{Gradshteyn}, $H\left(\varrho,a,b,q,p\right)$ can be derived as
\begin{align}\label{H_fun1}
H\left(\varrho,a,b,q,p\right)=&\frac{\varrho^{p+1}\Gamma\left(q+1\right)}{\left(k-p\right)\left(\varrho b+a\right)^{q+1}}
\notag\\
&\times{}_2F_1\left(1,q+1;q-p+1;\frac{a}{\varrho b+a}\right)
.\end{align}

When $p<0$, using \cite[Eq. (8.19.2)]{DLMF}, $H\left(\varrho,a,b,q,p\right)$ is obtained as
\begin{align}
H\left(\varrho,a,b,q,p\right)=\varrho^{p+1}\intop_{0}^{\infty}x^{q}\exp\left(-bx\right)E_{\left(-p\right)}\left(a\varrho x\right)dx
.\end{align}

Then, using \cite[Eq. (8.19.25)]{DLMF}, $H\left(\varrho,a,b,q,p\right)$ is derived as
\begin{align}\label{H_fun_abpq}
H\left(\varrho,a,b,q,p\right)=&\frac{\varrho^{p+1}\Gamma\left(q+1\right)}{\left(k-p\right)\left(\varrho b+a\right)^{q+1}}
\notag\\
&\times
{}_2F_1\left(1,q+1;q-p+1;\frac{a}{\varrho b+a}\right)
,\end{align}
which is as same as \eqref{H_fun1}.

Thus, substituting  \eqref{H_fun_abpq} into \eqref{six_group_sop1}, $\mathrm{SOP}_{1}^{\mathrm{L}}$ is derived as \eqref{sop1_final}, which yields \emph{Theorem \ref{SOP1_theo}}.

\section{}
The area of the target space shown in Fig. \ref{RS}, namely, the area of the
spherical doom with the radius $r_\mathrm{D}$, can be expressed as
\begin{align}
A_\mathrm{D}=2\pi r_\mathrm{D}^{2}\left(1-\cos\Psi_\mathrm{D}\right).
\end{align}

Thus, employing the method adopted in Appendix A of \cite{Pan_3D_VLC_integration}, and
considering that D is uniformly distributed, the CDF of $\theta_\mathrm{D}$ can be given as
\begin{align}
F_{\theta_\mathrm{D}}\left(x\right)&=\intop_{0}^{2\pi}\intop_{0}^{x}\frac{1}{A_\mathrm{D}}r_\mathrm{D}^{2}\sin\theta d\theta d\psi_\mathrm{D}
\notag\\
&
=2\pi\frac{r_\mathrm{D}^{2}}{A_{D}}\left(1-\cos x\right)
\notag\\
&
=\begin{cases}
\frac{1-\cos x}{1-\cos\Psi_\mathrm{D}}, &\mathrm{if}~0\leq x\leq\Psi_\mathrm{D};\\
0, & \mathrm{else}
\end{cases}.
\end{align}

So, it is easy to have the PDF of $\theta_\mathrm{D}$ as
\begin{align}
f_{\theta_\mathrm{D}}\left(x\right)=\frac{\sin x}{1-\cos\Psi_\mathrm{D}}.
\end{align}

To facilitate the following derivation, we denote
$w=d_\mathrm{D}^{2}=r_\mathrm{D}^{2}+H_\mathrm{R}^{2}-2r_\mathrm{D}H_\mathrm{R}\cos\theta_\mathrm{D}$, and the relationship between $\theta_\mathrm{D}$ and $w$ is
\begin{align}
h(\theta_{D})=\arccos\frac{r_{D}^{2}+H_{R}^{2}-w}{2r_{D}H_{R}}
.
\end{align}

Then, it is easy to write the PDF of $d_\mathrm{D}^{2}$ as
\begin{align}
&f_{d_\mathrm{D}^{2}}\left(w\right)=f_{\theta_\mathrm{D}}\left(w\right)\left|h^{'}(\theta_\mathrm{S})\right|
\notag\\
&=\frac{\sin\left(\arccos\frac{r_\mathrm{D}^{2}+H_\mathrm{R}^{2}-w}{2r_\mathrm{D}H_\mathrm{R}}\right)}{1-\cos\Psi_\mathrm{D}}
\frac{1}{2r_\mathrm{D}H_\mathrm{R}\sqrt{1-\left(\frac{r_\mathrm{D}^{2}+H_\mathrm{R}^{2}-w}{2r_\mathrm{D}H_\mathrm{R}}\right)^{2}}}
\notag\\
&=\frac{1}{2r_\mathrm{D}H_\mathrm{R}\left(1-\cos\Psi_\mathrm{D}\right)}.
\end{align}
\section{}

Combing \eqref{CDF_FSO} and \eqref{op2_define}, $\mathrm{OP_{2}}$ is expressed as
\begin{align}
\widetilde{\mathrm{OP_{2}}}=IG_{r+1,3r+1}^{3r,1}\left[\rho\gamma_{\mathrm{out}}|\begin{array}{c}
1,K_{1}\\
K_{2},0
\end{array}\right]
.
\end{align}

Substituting \eqref{rs_ave}, $\widetilde{\mathrm{OP_{2}}}$ is obtained as
\begin{align}
\widetilde{\mathrm{OP_{2}}}&=IG_{r+1,3r+1}^{3r,1}\left[\frac{\left(hab\right)^{r}}{\frac{P_\mathrm{R}
\zeta^{2}\mathcal{L}_\mathrm{r}^{2}}{\mathcal{L}_{\mathrm{FS}}\sigma_\mathrm{d}^{2}}r^{2r}}\gamma_{\mathrm{out}}|\begin{array}{c}
1,K_{1}\\
K_{2},0
\end{array}\right]
\notag\\
&=IG_{r+1,3r+1}^{3r,1}\left[d_\mathrm{D}^{2}\frac{\left({4\pi f_c}\right)^{2}\sigma_\mathrm{d}^{2}\left(hab\right)^{r}}{P_\mathrm{R}{c}^2\zeta^{2}\mathcal{L}_\mathrm{r}^{2}r^{2r}}
\gamma_{\mathrm{out}}|\begin{array}{c}
1,K_{1}\\
K_{2},0
\end{array}\right]
.
\end{align}

To facilitate the following analysis, we also denote
$\epsilon=\frac{\left({4\pi f_\mathrm c}\right)^{2}\sigma_{d}^{2}\left(hab\right)^{r}}{P_\mathrm{R}c^2\zeta^{2}\mathcal{L}_\mathrm{r}^{2}r^{2r}}\gamma_{\mathrm{out}}$,
$Z=d_\mathrm{D}^{2}$ and the PDF of $G=Z\epsilon$ is

\begin{align}
f_{\epsilon d_{D}^{2}}\left(g\right)=\frac{1}{2\epsilon r_{D}H_{R}\left(1-\cos\Psi_{D}\right)},
\end{align}
and the range of $g$ is $\epsilon d_{\mathrm D,\mathrm{min}}^{2}\leq g\leq\epsilon d_{\mathrm D,\mathrm{max}}^{2}$.

Thus, $\widetilde{\mathrm{OP_{2}}}$ can be given as
\begin{align}
&\widetilde{\mathrm{OP_{2}}}=\intop_{\epsilon d_{\mathrm D,\mathrm{min}}^{2}}^{\epsilon d_{\mathrm D,\mathrm{max}}^{2}}IG_{r+1,3r+1}^{3r,1}\left[g\Big|\begin{array}{c}
1,K_{1}\\
K_{2},0
\end{array}\right]f_{\epsilon d_{D}^{2}}\left(g\right)dg
\notag\\&
=\frac{I}{2 \epsilon r_\mathrm{D}H_\mathrm{R} \left(1-\cos\Psi_\mathrm{D}\right)}\intop_{\epsilon d_{\mathrm D,\mathrm{min}}^{2}}^{\epsilon d_{\mathrm D,\mathrm{max}}^{2}}G_{r+1,3r+1}^{3r,1}\left[g\Big|\begin{array}{c}
1,K_{1}\\
K_{2},0
\end{array}\right]dg
.
\end{align}

Utilizing \cite[Eq. (9.301)]{Gradshteyn}, the closed-form analytical expression of $\widetilde{\mathrm{OP_{2}}}$ is obtained as \eqref{op2_final}.

\section{}
When the height of R is unknown and D  is uniformly distributed in the space
shown in Fig. \ref{RS}, $H_\mathrm R$, $d_{\mathrm D,\mathrm{min}}^{2}$ and $d_{\mathrm D,\mathrm{max}}^{2}$  in \eqref{op2_final} are variables.
  Considering the randomness of the positions of R in Fig. \ref{SR} and D in Fig. \ref{RS},
  the accurate expression of ${\rm{OP}_2} $ can be acquired by integration, which is too complex to get a closed-form expression.
 But the ranges of $H_\mathrm R$, $d_{\mathrm D,\mathrm{min}}^{2}$ and $d_{\mathrm D,\mathrm{max}}^{2}$ are relatively very small compared to the values of them. Thus, considering the randomness of the positions of R
in Fig. \ref{SR} and D in Fig. \ref{RS},
the approximation of ${\rm{OP}_2} $ can be obtained by substituting  $H_\mathrm R$, $d_{\mathrm D,\mathrm{min}}^{2}$ and $d_{\mathrm D,\mathrm{max}}^{2}$  in their medians.

\begin{figure}
\setlength{\abovecaptionskip}{0pt}
\setlength{\belowcaptionskip}{10pt}
\centering
\includegraphics[width=2.4 in]{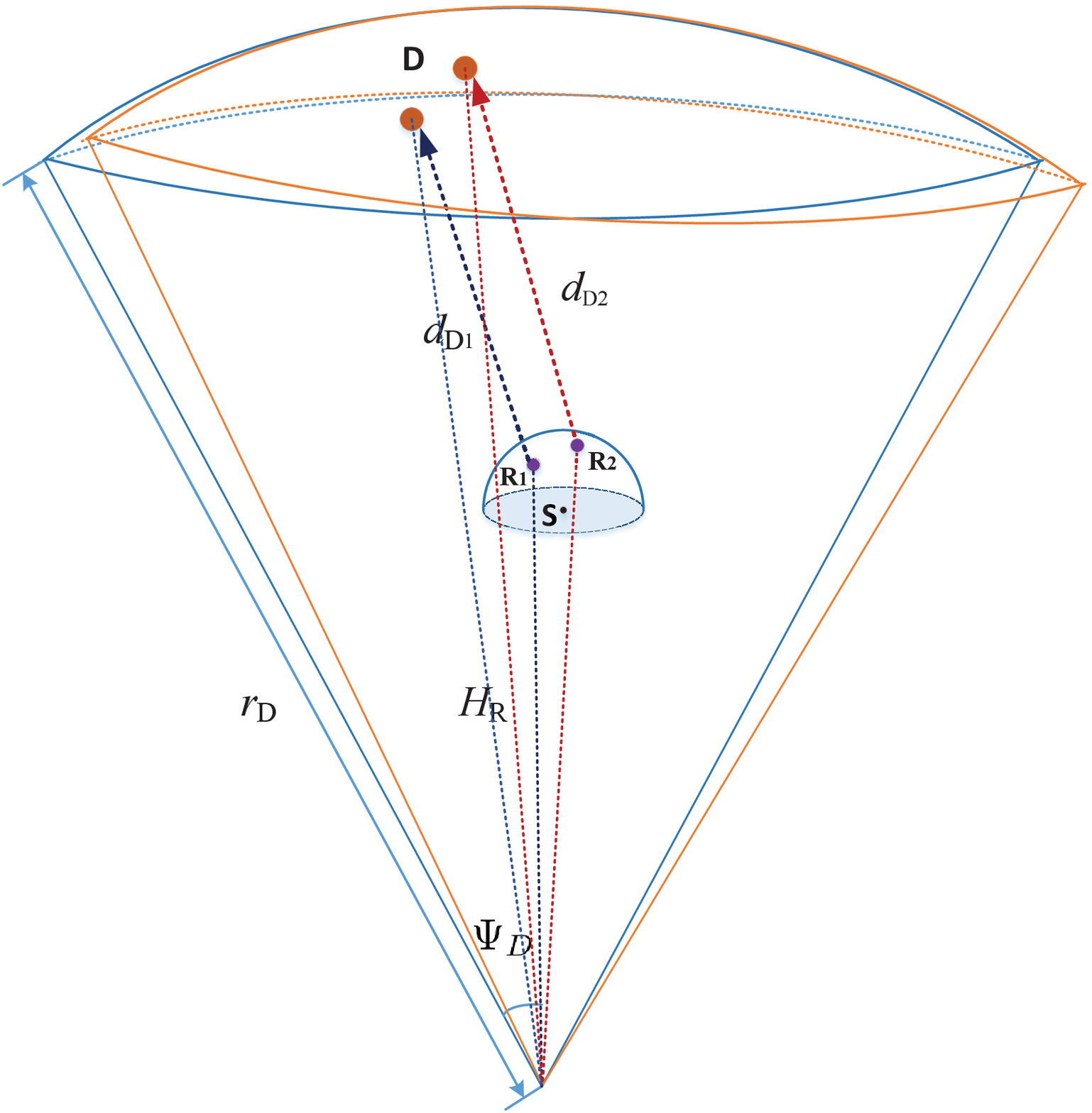}
\caption{The impact of R's height on $\mathrm{SOP}_{1}$ and $\mathrm{OP_{2}}$}
\label{compare}
\end{figure}

Fig. \ref{compare} gives an explanation for the approximation of $\mathrm{OP_{2}}$.  Own to the randomness of R, the height of R is in the range of $H_{\mathrm{min}}\le h_\mathrm{R}\le R_\mathrm{S}$, so  $H_\mathrm R$, $d_{\mathrm D,\mathrm{min}}$ and $d_{\mathrm D,\mathrm{max}}$ in \eqref{op2_final} are not fixed. The range of $H_\mathrm R$,  $d_{\mathrm D,\mathrm{min}}^{2}$ and $d_{\mathrm D,\mathrm{max}}^{2}$ are
$R_{\mathrm{earth}}+H_{\mathrm{min}}\leq H_\mathrm{R}\leq R_{\mathrm{earth}}+R_\mathrm{S}$,
$\ensuremath{\left(H_\mathrm{D}-R_\mathrm{S}\right)^{2}\text{\ensuremath{\le}}d_{\mathrm D,\mathrm{min}}^{2}\text{\ensuremath{\le}}
\left(H_\mathrm{D}-H_{\mathrm{min}}\right)^{2}}$
 and
 $r_\mathrm{D}^{2}+\left(R_{\mathrm{earth}}+R_\mathrm{S}\right)^{2}-2r_\mathrm{D}\left(R_{\mathrm{earth}}+R_\mathrm{S}\right)\cos\Psi_\mathrm{D}\leq d_{\mathrm D,\mathrm{max}}^{2}\leq r_\mathrm{D}^{2}+\left(R_{\mathrm{earth}}+H_{\mathrm{min}}\right)^{2}-2r_\mathrm{D}\left(R_{\mathrm{earth}}+H_{\mathrm{min}}
 \right)\cos\Psi_\mathrm{D}$, respectively.

Generally, the height of the low Earth orbit satellite is on the order of hundreds of kilometers and the height of the geostationary Earth orbit satellite is more than thirty thousand kilometers. Thus, compared to the height of the satellite, the height of R and the coverage space of S are quite small. Thus, the variation of  $d_{\mathrm D,\mathrm{min}}^{2}$  can be omitted. Similarly, the radius of the Earth is much longer than the height of R or the coverage space of S, the variation of  $d_{\mathrm D,\mathrm{max}}^{2}$ can also be ignored. Therefore, it is reasonable to take the medians of $H_\mathrm R$, $d_{\mathrm D,\mathrm{min}}^{2}$ and $d_{\mathrm D,\mathrm{max}}^{2}$  as their value to get the approximation of ${\rm{OP}_2} $, while considering the randomness of the positions of R in Fig. \ref{SR} and D in Fig. \ref{RS}.

We give a sample with $H_\mathrm{D}=550$ km, $R_\mathrm{S}=150$ m, $H_{\mathrm{min}}=80$ m, $\Psi_\mathrm{D}=\frac{\pi}{6}$.
 The rate of change for  $H_\mathrm R$, ${d}_{\mathrm D,\mathrm{min}}^{2}$ and ${d}_{\mathrm D,\mathrm{max}}^{2}$
are
$\Delta_{0}=\frac{\max\left\{ H_\mathrm{R}\right\} -\min\left\{ H_\mathrm{R}\right\} }{\min\left\{ H_\mathrm{R}\right\} }=0.00001$,
$\Delta_{1}=\frac{\max\left\{ d_{\mathrm D,\mathrm{min}}^{2}\right\} -\min\left\{ d_{\mathrm D,\mathrm{min}}^{2}\right\} }{\min\left\{ d_{\mathrm D,\mathrm{min}}^{2}\right\} }=0.00025$, and
$\Delta_{2}=\frac{\max\left\{ d_{\mathrm D,\mathrm{max}}^{2}\right\} -\min\left\{ d_{\mathrm D,\mathrm{max}}^{2}\right\} }{\min\left\{ d_{\mathrm D,\mathrm{max}}^{2}\right\} }=0.000004$, respectively. Thus, the approximation of ${\rm{OP}_2} $ is reasonable.

\section{}
Based on the above analysis, $\mathrm{SOP}_{1}$ and $\mathrm{OP_{2}}$ are not independent because the randomness of R leads to the variation of $d_\mathrm R$  in \eqref{cdf_dR},  $d_\mathrm D^2$ in \eqref{pdf_dSR2} and  the  variation of $H_\mathrm R$, $d_{\mathrm D,\mathrm{min}}^{2}$ and $d_{\mathrm D,\mathrm{max}}^{2}$ in \eqref{op2_final}. When we adopt $\bar H_\mathrm R$, $\bar d_{\mathrm D,\mathrm{min}}^{2}$ and $\bar d_{\mathrm D,\mathrm{max}}^{2}$ in \eqref{op2_final_app}, we get rid of the correlation of $\mathrm{SOP}_{1}$ and $\mathrm{OP_{2}}$ with sufficient reasons.
This also implies that the randomness of R has some effects on  $\mathrm{SOP}_{1}$ but has little effect on $\mathrm{OP_{2}}$. Thus, $\mathrm{SOP}_{1}$ and $\mathrm{OP_{2}}$ can be viewed
as the independent variable for simplification.


\begin{IEEEbiography}
[{\includegraphics[width=1in,height=1.25in,clip,keepaspectratio]{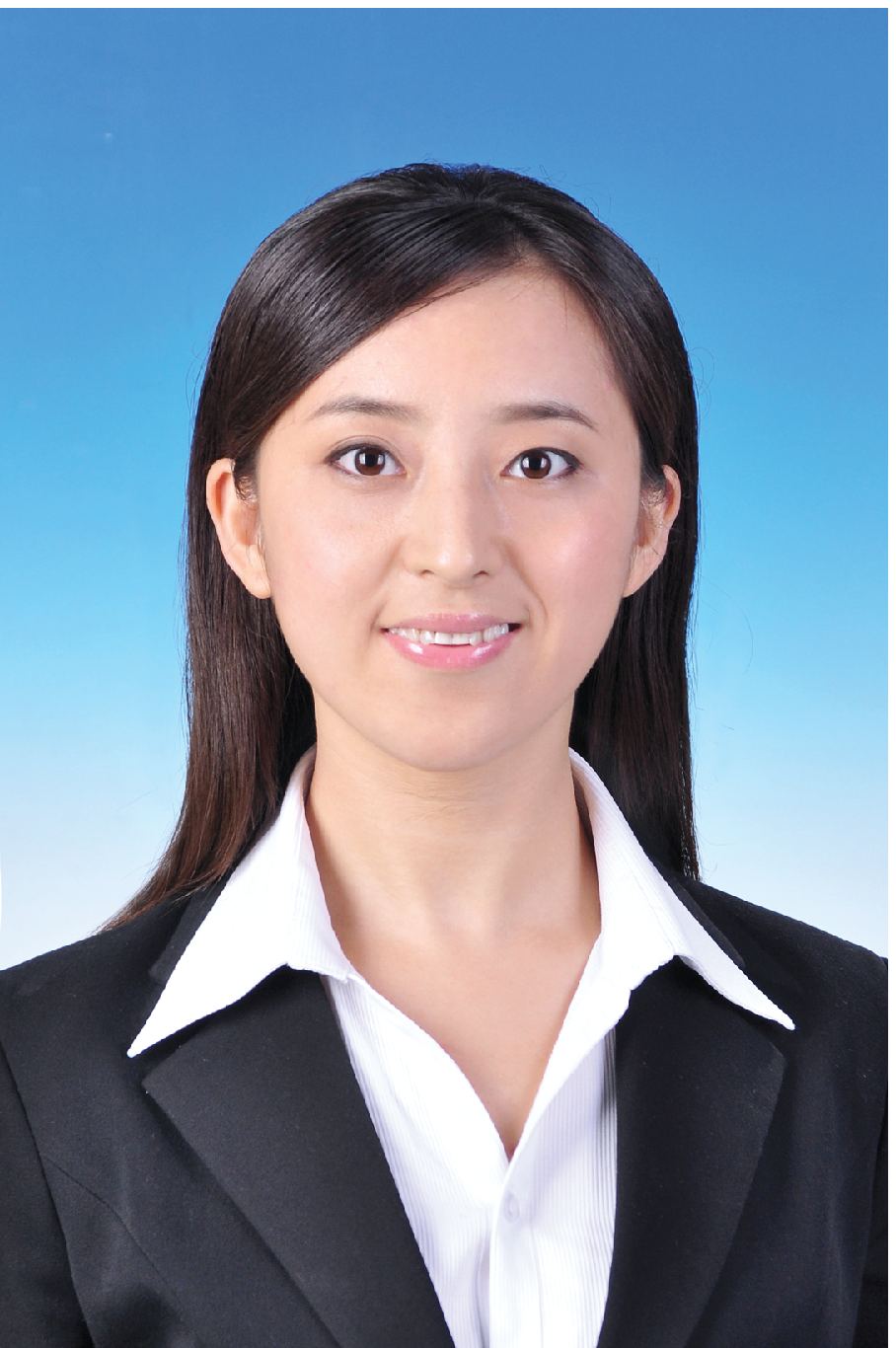}}]{Yuanyuan Ma}
 received the B.Sc  degree in communication engineering from Henan Normal University, China, in 2011, and the M.S. degree in information and communication engineering from Beijing Institute of Technology, China, in 2014. She is currently pursuing the Ph.D. degree in information and communication engineering from Beijing University of Posts and Telecommunications, China. Her current research interests include satellite communication, cognitive radio networks and performance analysis.
\end{IEEEbiography}

\begin{IEEEbiography}[{\includegraphics[width=1in,height=1.25in,clip,keepaspectratio]{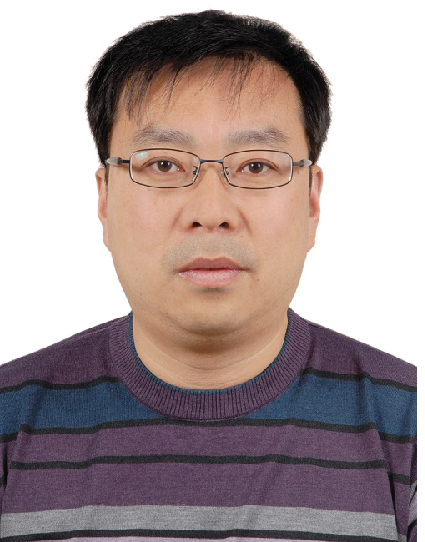}}]{Tiejun Lv}
(M'08-SM'12) received the M.S. and Ph.D. degrees in electronic engineering from the University of Electronic Science and Technology of China (UESTC), Chengdu, China, in 1997 and 2000, respectively. From January 2001 to January 2003, he was a Postdoctoral Fellow with Tsinghua University, Beijing, China. In 2005, he was promoted to a Full Professor with the School of Information and Communication Engineering, Beijing University of Posts and Telecommunications (BUPT). From September 2008 to March 2009, he was a Visiting Professor with the Department of Electrical Engineering, Stanford University, Stanford, CA, USA. He is the author of three books, more than 100 published IEEE journal papers and 200 conference papers on the physical layer of wireless mobile communications. His current research interests include signal processing, communications theory and networking. He was the recipient of the Program for New Century Excellent Talents in University Award from the Ministry of Education, China, in 2006. He received the Nature Science Award in the Ministry of Education of China for the hierarchical cooperative communication theory and technologies in 2015.
\end{IEEEbiography}

\begin{IEEEbiography}
[{\includegraphics[width=1in,height=1.25in,clip,keepaspectratio]{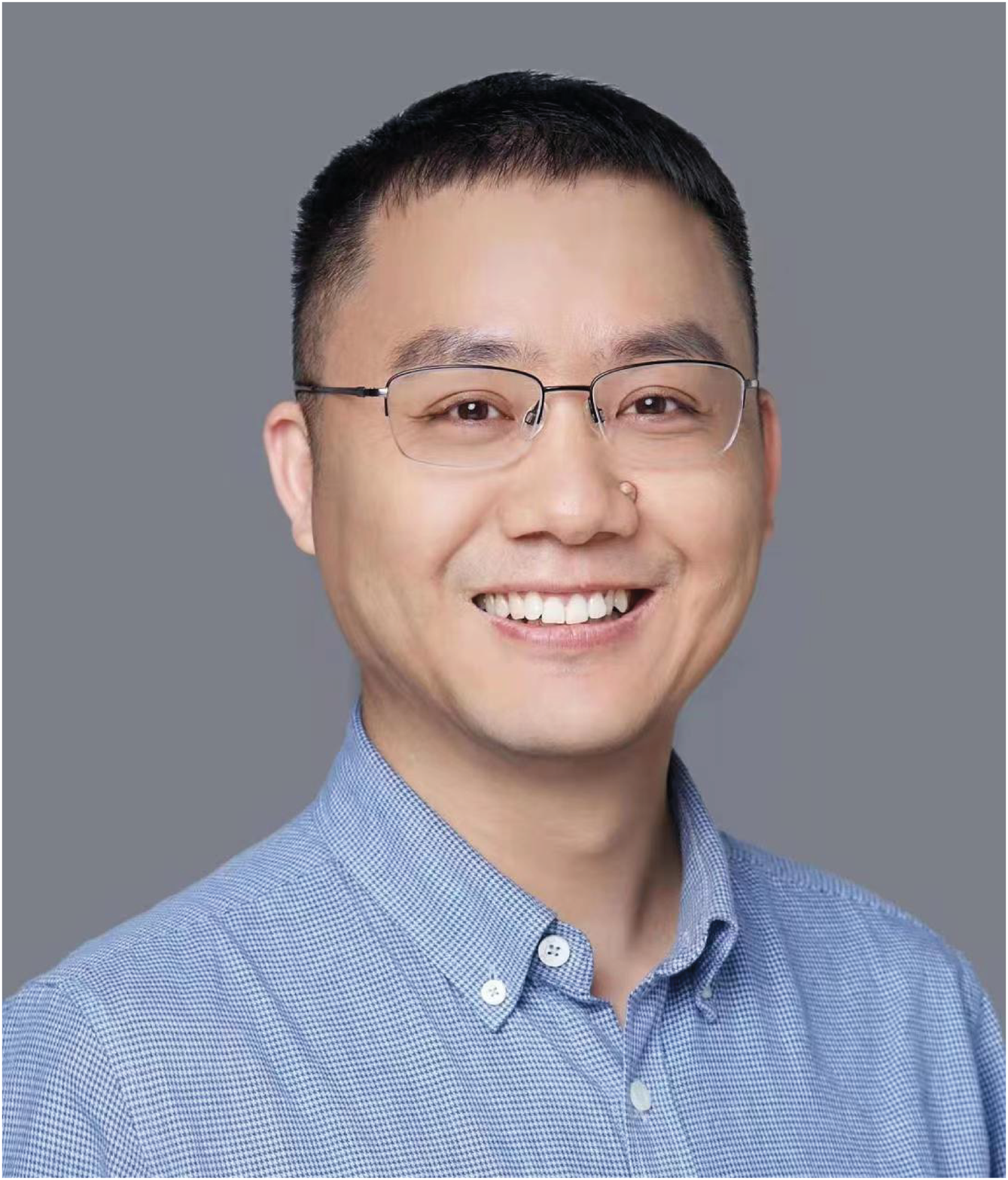}}]{Gaofeng Pan}
(Senior Member, IEEE) received
his B.Sc in Communication Engineering from
Zhengzhou University, Zhengzhou, China, in 2005,
and the Ph.D. degree in Communication and Information Systems from Southwest Jiaotong University,
Chengdu, China, in 2011. He is currently with the
School of Cyberspace Science and Technology, Beijing Institute of Technology, China, as a Professor.
His research interest spans special topics in communications theory, signal processing, and protocol
design.
\end{IEEEbiography}

\begin{IEEEbiography}
[{\includegraphics[width=1in,height=1.25in,clip,keepaspectratio]{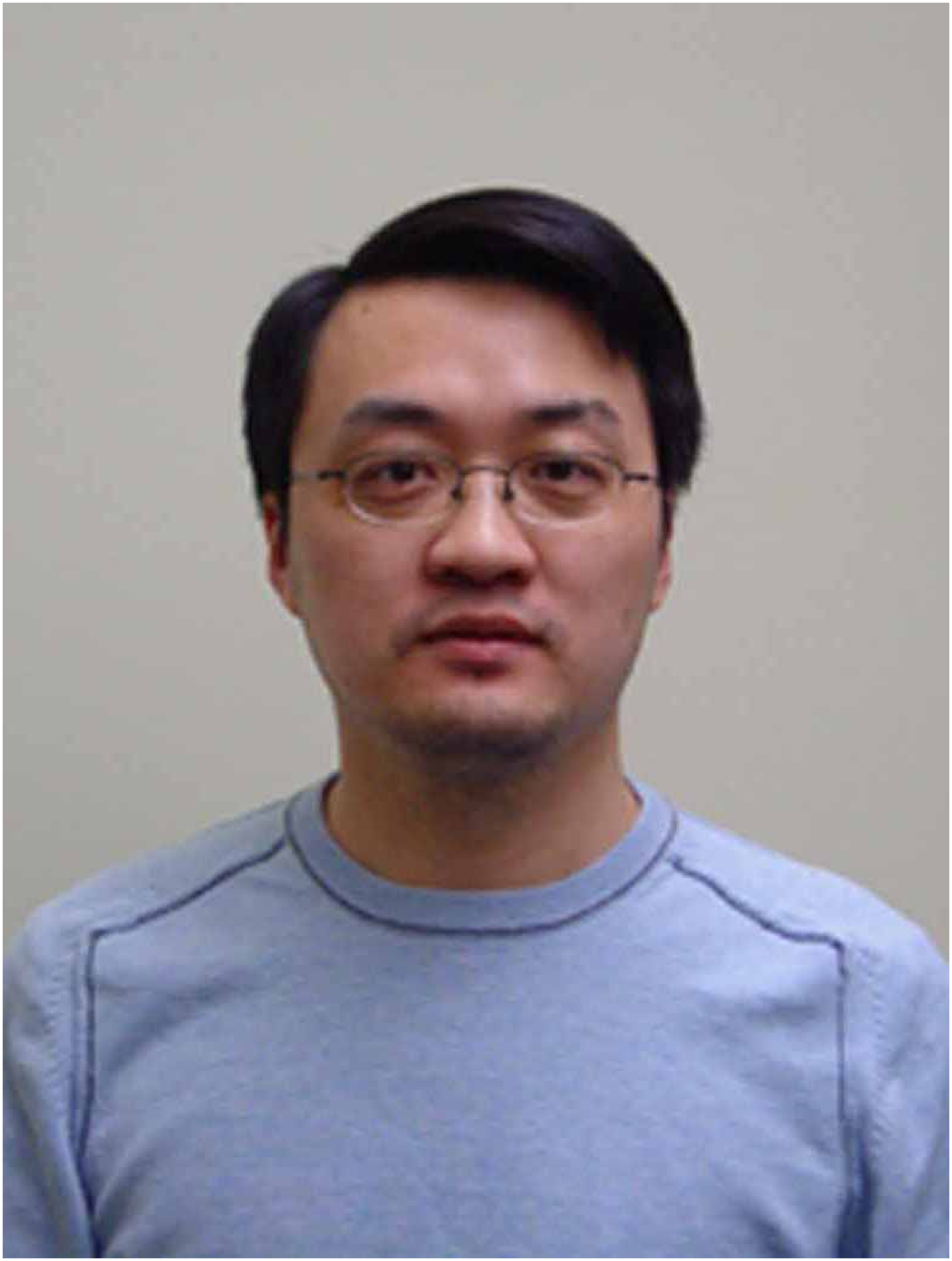}}]{Yunfei Chen}
(S’02-M’06-SM’10) received his B.E. and M.E. degrees in electronics engineering from Shanghai Jiaotong University, Shanghai, P.R.China, in 1998 and 2001, respectively. He received his Ph.D. degree from the University of Alberta in 2006. He is currently working as an Associate Professor at the University of Warwick, U.K. His research interests include wireless communications, cognitive radios, wireless relaying and energy harvesting.
\end{IEEEbiography}

\begin{IEEEbiography}[{\includegraphics[width=1in,height=1.25in,clip,keepaspectratio]{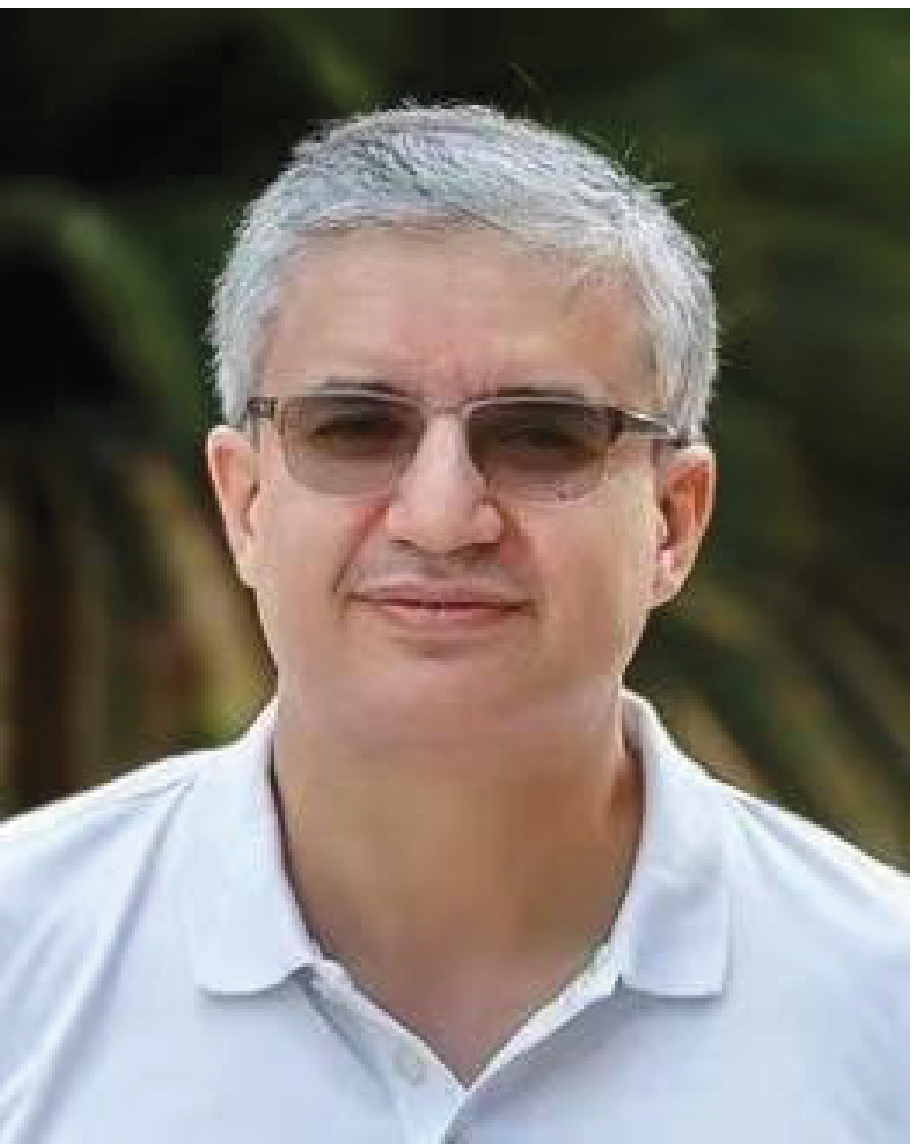}}]{Mohamed-Slim Alouini }(Fellow, IEEE) was born in Tunis, Tunisia. He received the Ph.D. degree in Electrical Engineering from the California Institute of Technology (Caltech), Pasadena, CA, USA, in 1998. He served as a faculty member in the University of Minnesota, Minneapolis, MN, USA, then in the Texas A\&M University at Qatar, Education City, Doha, Qatar before joining King Abdullah University of Science and Technology (KAUST), Thuwal, Makkah Province, Saudi Arabia as a Professor of Electrical Engineering in 2009. His current research interests include modeling, design, and performance analysis of wireless communication systems.
\end{IEEEbiography}

\end{document}